\documentclass[10pt, twocolumn]{IEEEtran}
\usepackage{etoolbox}

\usepackage[keeplastbox]{flushend}  

\usepackage{cite}
\usepackage[pdftex]{graphicx}
\graphicspath{{../pdf/}{../jpeg/}}
\DeclareGraphicsExtensions{.pdf,.jpeg,.png}

\usepackage{amsmath}
\allowdisplaybreaks[4]
\usepackage{amsthm, amssymb}
\newtheorem{prop}{Proposition}

\usepackage{array}
\usepackage{stfloats}
\usepackage{mathrsfs}
\usepackage{breqn}

\usepackage[linesnumbered,ruled,vlined]{algorithm2e}
\usepackage{enumerate}

\usepackage{color} 

\begin{document}

\begin{titlepage}

\null\vfill
\noindent This article has been accepted by IEEE Transactions on Vehicular Technology.

\vspace{2em}

\noindent Copyright (c) 2015 IEEE. Personal use of this material is permitted. However, permission to use this material for any other purposes must be obtained from the IEEE by sending a request to pubs-permissions@ieee.org.
\vspace{2em}

\noindent Z. Ye, C. Pan, H. Zhu and J. Wang,``Tradeoff Caching Strategy of Outage Probability and Fronthaul Usage in Cloud-RAN,'' in \emph{IEEE Transactions on Vehicular Technology}, vol. PP, no. 99, pp. 1-1.

\vspace{2em}
\noindent doi: 10.1109/TVT.2018.2797957

\vspace{2em}
\noindent \textcolor{blue}{https://doi.org/10.1109/TVT.2018.2797957}
\vfill\vfill\vfill\vfill
\end{titlepage}

\title{Tradeoff Caching Strategy of Outage Probability and Fronthaul Usage in Cloud-RAN}

\author{Zhun Ye,~\IEEEmembership{Member,~IEEE}, Cunhua Pan,~\IEEEmembership{Member,~IEEE}, 
	
	Huiling Zhu,~\IEEEmembership{Senior Member,~IEEE}, and	Jiangzhou Wang,~\IEEEmembership{Fellow,~IEEE}
	
	\thanks{Copyright (c) 2015 IEEE. Personal use of this material is permitted. However, permission to use this material for any other purposes must be obtained from the IEEE by sending a request to pubs-permissions@ieee.org.
		
		Part of this work was presented in IEEE International Conference on Communications (ICC), Paris, 2017 \cite{myICC}.
		
		This work was supported by the China Scholarship Council (CSC), the Fundamental Research Funds of Shandong University (No. 2015ZQXM008), a Marie Curie International Outgoing Fellowship within the 7th European Community Framework Programme under the Grant PIOFGA-2013-630058 (CODEC), the UK Engineering and Physical Sciences Research Council under the Project EP/L026031/1 (NIRVANA), and the Framework of EU Horizon	2020 Programme under the Grant 644526 (iCIRRUS). 
		
	    Zhun Ye is with the School of Mechanical, Electrical and Information Engineering, Shandong University, Weihai, Shandong Province, 264209, P. R. China, e-mail: zhunye@sdu.edu.cn
	
	    Cunhua Pan is with the School of Electronic Engineering and Computer Science, Queen Mary University of London, London, E1 4NS, United Kingdom, Email: c.pan@qmul.ac.uk
		
        Huiling Zhu and Jiangzhou Wang are with the School of Engineering and Digital Arts, University of Kent, Canterbury, Kent, CT2 7NT, United Kingdom, e-mail: \{h.zhu, j.z.wang\}@kent.ac.uk}
}

{}

\maketitle

\begin{abstract}
	In this paper, tradeoff content caching strategy is proposed to jointly minimize the cell average outage probability and fronthaul usage in cloud radio access network (Cloud-RAN). At first, an accurate closed form expression of the outage probability conditioned on the user's location is presented, and the cell average outage probability is obtained through the composite Simpson's integration. The caching strategy for jointly optimizing the cell average outage probability and fronthaul usage is then formulated as a weighted sum minimization problem, which is a nonlinear 0-1 integer problem. Two heuristic algorithms are proposed to solve the problem. Firstly, a genetic algorithm (GA) based approach is proposed. Numerical results show that the performance of the proposed GA-based approach with significantly reduced computational complexity is close to the optimal performance achieved by exhaustive search based caching strategy, and the GA-based approach can improve the performance by up to 47.5\% on average than the typical probabilistic caching strategy. Secondly, in order to further reduce the computational complexity, a mode selection approach is proposed. Numerical results show that this approach can achieve near-optimal performance over a wide range of the weighting factors through a single computation.	
\end{abstract}

\begin{IEEEkeywords}
Caching strategy, Cloud-RAN, joint optimization, outage probability, fronthaul usage.
\end{IEEEkeywords}

\IEEEpeerreviewmaketitle

\section{Introduction}

\IEEEPARstart{T}HE combination of network densification and coordinated multipoint transmission is a major technical trend in the fifth generation (5G) wireless mobile systems to improve the overall system performance \cite{EE_5G, 5G_survey2, trend_5G, green_hetnet}. In the traditional radio access network (RAN) architecture, each cell has its own base station (BS), where the radio functionality is statically assigned to the base band processing module. Adding more base stations and introducing multiple input multiple output (MIMO) technology will increase the complexity of the network and result in  higher total cost of ownership (TCO) for the mobile operators \cite{cloud_overview, centralization}.

Consisting of centralized base band processing resources, known as base band unit (BBU) pool, and distributed remote radio heads (RRHs) or remote antenna units \cite{comparison, train, allocation_DAS}, cloud radio access network (Cloud-RAN) becomes a new type of RAN architecture to support multipoint transmission and access point densification required by 5G systems\cite{cloud_overview, centralization,green_mimo,incomplete}. The scalable, virtualized, and centralized BBU pool is shared among cell sites, and its computing resources can be dynamically allocated to different cells according to their traffic. The RRHs are responsible for the radio processing task, and they are connected with the BBU pool through fronthauls, while the BBU pool performs the base band processing task and it is connected to the core network through backhauls. Thanks to the novel architecture, Cloud-RAN has many advantages such as cost effective, lower energy consumption, higher spectral efficiency, scalability and flexibility etc., which makes itself a promising candidate for the 5G deployment. Particularly, Cloud-RAN is also a competitive solution for the heterogeneous vehicular networks, which can provide better quality of service (QoS) to intense vehicular users in an urban environment \cite{game_vt, 5G_vt}. However, existing fronthaul/backhaul of Cloud-RAN cannot meet the requirements of the emerging huge data and signaling traffic in terms of transmission bandwidth requirements, stringent latency constraints and energy consumption etc. \cite{backhaul_challenge, redesigning, wireless_backhauling}, which has become the bottleneck of the evolution towards 5G.

Statistics showed that a large amount of the data traffic is generated from a small amount of most popular content files. 
These popular files are requested by a large amount of users, which results in duplicated transmissions of the same content files on the fronthaul and backhaul links. Therefore, content caching in RAN can be a promising solution to significantly reduce the fronthaul/backhaul traffic \cite{caching_air, caching_multicast,fundamental}. During off-peak times, popular content files can be transferred to the cache-enabled access points (macro base station, small cell, relay node etc.). 
If the files requested by mobile users are cached in the access points of the RAN, the files will be transmitted directly from the RAN's cache without being fetched from the core network, which can significantly reduce the fronthaul/backhaul traffic and meanwhile shorten the access latency of the files, thus improve users' quality of experience (QoE). In Cloud-RAN, thanks to the ongoing evolution of fronthaul technology and function splitting between the BBU and RRHs \cite{redesigning, split}, there comes possibility to realize content caching in RRHs, which allows users fetching required content files directly from RRHs and thus can further reduce fronthaul traffic.

There are two stages related with caching: \emph{caching placement stage} and \emph{caching delivery stage} \cite{fundamental}. Caching placement, or known as caching strategy, is the stage to determine which files should be stored in which cache-enabled access points, and delivery stage refers to stage of transmitting the requested files from access points to mobile users through wireless channels. Among these two stages, caching placement is performed for a relative long-timescale. Once a caching placement is carried out, it will not change very frequently. The reason is that the popularity of the content files will remain the same for a relative long period such as several hours, one day, or even longer time. On the other hand, delivery stage runs in a short-timescale. In a delivery stage, the wireless transmission scheme should be able to adapt to the instantaneous channel state information (CSI) which varies very rapidly.

There are many researches investigating the delivery stage with the target of data association or/and energy consumption optimization under a given caching strategy, such as \cite{delivery_beamforming, delivery_content_journal, match_cache}. On the other hand, caching strategy is of importance because it is the initial step to perform caching and obviously it will have an impact on the performance of the delivery stage. The researches investigating  caching strategies generally focus on reducing the file access latency \cite{letter, delay,prob_cl}, or minimizing the transmission cost of the backhaul \cite{optimizing_cl,pso}, or both of them \cite{hyper}.  However, the wireless transmission characteristics such as fading were not considered in the aforementioned researches, i.e., it was assumed that the wireless transmission is error-free. The caching strategy will affect the wireless transmission performance such as outage probability, which is an important metric of the system's performance. For the fronthaul/backhaul traffic or average file access delay reduction, caching different files in the RAN will be optimal, however there is no transmit diversity to combat fading in the file delivery stage, which may decrease the reliability of the wireless transmission. Hence, caching strategy should be optimized by taking into consideration the wireless transmission performance. 

There are some papers considering wireless fading characteristics when designing caching strategy \cite{helpers,backhaul-aware}.  The authors only considered small scale Rayleigh fading by assuming that the user has the same large scale fading at any location. However, in reality, several RRHs will jointly serve the user in Cloud-RAN, and obviously the distance between each RRH and the user will not be the same, so it is important to consider large scale fading in wireless transmission. In addition, they focused on single-objective optimization without considering the fronthaul/backhaul usage.

The aim of caching in RRHs of Cloud-RAN is to significantly reduce the fronthaul traffic. Fronthaul usage, i.e., whether the fronthaul is used, is a metric which can reflect not only the file delivery latency but also the energy consumption of the fronthaul. For example, lower fronthaul usage implies there are more possibilities that mobile user can access the content files in near RRHs, which will shorten the file access latency, meanwhile the fronthaul cost (i.e., the energy consumption) will be lower. On the other hand, outage probability is an important performance metric of the system, which reflects the reliability of the wireless transmission, i.e., whether the requested content files can be successfully transferred to the user, and it also reflects the utility of the wireless resources. 
If replicas of certain content files are cached in several RRHs, the outage probability will be reduced due to the transmit diversity in wireless transmissions, while the fronthaul usage will become higher because the total number of different files cached in the RRHs are reduced and there is a high possibility to fetch files from the BBU pool. On the other hand, caching different files in the RRHs will reduce the fronthaul usage, while the outage probability will become relatively higher due to the decrease of wireless diversity. Therefore, there exists tradeoff between fronthaul usage and outage probability.

In this paper, we investigate downlink transmission in a virtual cell in Cloud-RAN, such as a hot spot area, shopping mall, or an area covered by the Cloud-RAN based vehicular network etc. The tradeoff caching strategy is proposed to jointly minimize the cell average outage probability and the fronthaul usage. A realistic fading channel is adopted, which includes path loss and small scale Rayleigh fading. The caching strategy is designed based on the long-term statistics about the users' locations and content file request profiles. The major contributions of this paper are: 
\begin{enumerate}[1)]
	\item Closed form expression of outage probability conditioned on the user's location is derived, and the cell average outage probability is obtained through the composite Simpson's integration. Simulation results show that the analysis is highly accurate.
	
	\item{ The joint optimization problem is formulated as a weighted sum minimization of cell average outage probability and fronthaul usage, which is a 0-1 integer problem. Two heuristic algorithms are proposed to solve the problem:
		\begin{enumerate}[a)]
			\item An effective genetic algorithm (GA) based approach is proposed, which can achieve nearly the same performance as the optimal exhaustive search, while the computational complexity is significantly reduced. 
			\item In order to further reduce the computational complexity, a mode selection approach is proposed. Simulation results show that it can achieve near-optimal performance over a wide range of weighting factors through a single computation.
	\end{enumerate}}
\end{enumerate}

The remainder of this paper is organized as follows. Section \ref{related_works} reviews the related works. System model is described in Section \ref{system_model}. The optimization problem is formulated in Section \ref{problem_formulation} and the cell average outage probability and fronthaul usage are analyzed. The proposed GA-based approach and the mode selection schemes are described in Section \ref{caching_placement}. Numerical  results are given in Section \ref{numerical_results} and the conclusion is given in Section \ref{conclusion}.   

{\emph{Notations}}: $ \mathbb{E}(\cdot) $ denotes statistical expectation, and $ \operatorname{Re}(\cdot) $ denotes the real part of a complex number. $ \mathbf{A}^{L\times N}=\{a_{l,n}\} $ denotes $ L\times N$  matrix, $ a_{l,n} $ or $ \mathbf{A}(l,n) $ represents the $ (l,n) $-th entry of the matrix. $ \mathbb{R}^+ $ denotes the set of positive real numbers, 
and $ \mathbb{Z}^+ $ denotes positive integer set. 
$ \mathcal{CN}(\mu,\sigma^2) $ represents complex Normal distribution with mean $ \mu $ and variance $ \sigma^2 $, and $ \chi^2(k) $ is the central Chi-squared distribution with $ k $ degrees of freedom.

\section{Related Works} \label{related_works}

From the caching point of view, there exists significant similarities between Cloud-RAN, small cell network, marcocell network and some vehicular networks etc. There are many researches investigating the delivery stage under a certain caching strategy, and the main target was to optimize the data association (e.g., RRH clustering and transmit beamforming), such as \cite{delivery_beamforming, delivery_content_journal, match_cache}. In \cite{delivery_beamforming} and \cite{ delivery_content_journal}, optimal base station clustering and beamforming were investigated to reduce the backhaul cost and transmit power cost under certain caching strategy. In addition, the performances of different commonly used caching strategies, such as popularity-aware caching, random caching, and probabilistic caching, were compared in \cite{delivery_content_journal}. In \cite{match_cache}, assuming there are several small base stations in an orthogonal frequency division multiple access (OFDMA) macro cell, the optimal association of the users and small base stations was investigated to reduce the long-term backhaul bandwidth allocation. In these researches, the caching strategy was assumed to be fixed when designing the delivery schemes, which is because the delivery stage runs in a much shorter timescale than the caching placement stage.

On the other hand, caching strategy has attracted widely concern recently, the related researches mainly focused on the reduction of file access latency \cite{letter,delay,prob_cl}, fronthaul/backhaul transmissions \cite{optimizing_cl,pso}, or both of them \cite{hyper}. In \cite{letter}, a collaborative strategy of simultaneously caching in BS and mobile devices was proposed to reduce the latency for requesting content files. The proposed optimal strategy was to fill the BS's cache with the most popular files and then cache the remaining files of higher popularity in the mobile devices. In \cite{delay}, a distributed algorithm with polynomial-time and linear-space complexity was proposed to minimize the expected overall access delay in a cooperative cell caching scenario. The delay from different sources to the user was modeled as uniformly distributed random variables within a certain range. In \cite{prob_cl}, the probabilistic caching strategy was optimized in clustered cellular networks, where the limited storage capacity of the small cells and the amount of transferred contents within the cluster were considered as two constraints to minimize the average latency. The optimized caching probability of each content file was obtained. 

In \cite{optimizing_cl}, a coded caching placement was proposed to minimize the backhaul load in a small-cell network, where multicast was adopted. The file and cache sizes were assumed to be heterogeneours. In \cite{pso}, to minimize the total transmission cost among the BSs and from the core network, each BS's cache storage was divided into two parts, the first part of all the BSs cached same files with higher popularity ranks, while the second part of all the BSs stored different files. The cache size ratio of the two parts was optimized through particle swarm optimization (PSO) algorithm. In \cite{hyper}, caching strategy was investigated in a Cloud-RAN architecture based networks, and the average content provisioning cost (e.g., latency, bandwidth etc.) was analyzed and optimized subjecting to the sum storage capacity constraint. Analytical results of the optimal storage allocation (how to partition the storage capacity between the control BS and traffic BS) and cache placement (decision on which file to cache) were obtained. However, the aforementioned researches did not take wireless transmission characteristics into consideration. In practice, the caching strategy will have an impact on the performance of the delivery stage, so wireless transmission performance should be considered in order to optimize the caching strategy.

There are some papers considering wireless fading characteristics when designing/investigating caching strategy \cite{two_tier,prob_tvt,helpers,backhaul-aware}. Stochastic geometry was used to analyze large scale networks in \cite{two_tier, prob_tvt}. In \cite{two_tier}, considering a cache-enabled two-tier heterogeneous network with one macrocell BS and several small-cell BSs, outage probability, throughput, and energy efficiency (EE) were analyzed. Each of the BSs caches the most popular content files until the storage is full filled. Numerical results showed that larger small-cell cache capacity may leads to lower network energy efficiency when the density of the small cells is low. In \cite{prob_tvt}, the performance of probabilistic caching strategy was analyzed and optimized in a small-cell environment, and the aim was to maximize the successful download probability of the content files. However, only \emph{probabilistic content placement} can be obtained through using the tool of stochastic geometry \cite{tier_level}, that is, the probability of a certain file should be cached in the access points. In \cite{helpers}, optimal caching placement was obtained through a greedy algorithm to minimize the average bit error rate (BER) in a macro cell with many cache-enabled helpers and each helper can cache only one file. The user selects one helper with the highest instantaneous received signal to noise ratio (SNR) among the helpers which cache the requested file. If none of the helpers cache the requested file, the user will fetch the file from the BS. In \cite{backhaul-aware}, cache-enabled BSs are connected to a central controller via backhaul links. The aim was to minimize the average download delay. Similar to \cite{helpers}, the user selects the BS with the highest SNR in the candidate BSs caching the requested files. In \cite{helpers} and \cite{backhaul-aware}, the authors only considered small scale Rayleigh fading by assuming that the user has the same large scale fading at any location, which is unpractical. In addition, they focused on single-objective optimization without considering the fronthaul/backhaul usage.

Inspired by the aforementioned researches, in this paper, outage probability is used to reflect the wireless transmission performance, and fronthaul usage is used to reflect the transmission latency and power consumption etc. Outage probability and fronthaul usage are jointly considered when designing the caching strategy, which leverages the tradeoff between caching the same content files to obtain lower outage probability or caching different content files to reduce fronthaul usage. Considering the distances from each RRH to the user are different in a Cloud-RAN environment, a more practical fading channel model which includes both large and small scale fading is adopted.

\section{System Model} \label{system_model}

It is assumed that there are $ N $ cache-enabled RRHs in a circular cell with radius $R$, and the set of RRH cluster is denoted as $ \mathcal{N}=\{1,2,\cdots,N\}  $. The file library with a total of $ L $ content files is denoted as $ \mathcal{F}=\{F_1,F_2,\cdots,F_L\}  $, where $ F_l $ is the $ l $-th ranked file in terms of popularity, i.e., $ F_1 $ is the most popular content file.
The popularity distribution of the files follows the Zipf's law \cite{zipf}, and the request probability of the $ l $-th ranked content file is 
\begin{equation}
P_l=\dfrac{l^{-\beta}} {\sum\nolimits_{n=1}^{L}{n^{-\beta}}}~,
\end{equation}
where $ \beta \in \left[ 0,+\infty \right)  $ is the skewness factor. The popularity is uniformly distributed over content files when $ \beta = 0 ~(P_l=1/L, \forall l) $ and becomes more skewed towards the most popular files as $ \beta $ grows, while large popularity skewness is usually observed in wireless applications.

For simplicity, it is assumed that all content files have the same size, and the file size is normalized to 1. Even though the file size will not be equal in practice, each file can be segmented into equal-sized chunks for placement and delivery \cite{caching_multicast, FemtoCaching}. Considering the BBU pool can be equipped with sufficient storage space, it is assumed that all the $ L $ content files are cached in the BBU pool \footnote{Generally speaking, the backhaul connecting the BBU pool and the core network will have larger transmission bandwidth than the fronthaul, so only the fronthaul usage reduction is considered in this paper. In practice, the BBU pool can not cache all the content files originated in the Internet, however, if the requested file is not cached in the BBU pool, it can be fetched from the core network through using backhaul, then it is the same as the file is already cached in the BBU pool as we only focus on the fronthaul usage rather than backhaul usage.}. Some of the content files can be further cached in the RRHs in order to improve the system's performance, and a file can be cached in one or more RRHs depending on the caching strategy. The $ n $-th RRH can cache $ M_n $ files, and generally $\sum\nolimits_{n=1}^{N} M_n < L $. That is, the total caching storage space in all the RRHs is smaller than that in the BBU pool. The caching placement of the content files in the RRHs can be denoted by a binary placement matrix $ \mathbf{A}^{L\times N} $, with the $ (l,n) $-th entry
\begin{equation}
a_{l,n}=\left\{ 
\begin{array}{cl}
1,& \text{the $ n $-th RRH caches the $ l $-th file}\\
0,& \text{otherwise}\\	
\end{array}
\right.
\end{equation}
indicating whether the $ l $-th content file is cached in the $ n $-th RRH, and $ \sum_{l=1}^{L}a_{l,n} = M_n,~\forall n $.

Single user scenario is considered in this paper. However, the proposed algorithms can be applied in practical multiuser systems with orthogonal multiple access technique such as OFDMA system, in which each user is allocated with different subcarriers and there is no interference \cite{chunk,joint,radio}. It is assumed that the user can only request for one file at one time, and all the RRHs caching the requested file will serve the user. If none of the RRHs caches the requested file, the file will be transferred to all the RRHs from the BBU pool through fronthauls, and then to the user from all the RRHs through wireless channels. The service RRH set for the user with respect to (w.r.t.) the $ l $-th file is denoted as 
\begin{equation}
\Phi_l=\left\{
\begin{array}{cl}
\{n|a_{l,n}=1, n\in \mathcal{N}\} &,~\exists n $ such that $ a_{l,n}=1\\
\mathcal{N} &,~\text{$ a_{l,n}=0$ for $\forall n $}\\
\end{array}
\right.,
\end{equation} 
with cardinality $ |\Phi_l|\in \{1,2,\cdots,N\}$, ($ l\in\{1,2,\cdots,L\} $). The system model and file delivery scheme are illustrated in \figurename{\ref{system model}}. For example, when the user requests for the $ l_1 $-th file which is not cached in any of the RRHs, the file will be transferred from the BBU pool to all the RRHs through fronthauls and then transmitted to the user. Then the user's service RRH set is $ \Phi_{l_1}=\{1,2,3,4\} $. When the user requests for the $ l_2 $-th file which is already cached in RRH 2 and RRH 3 via caching placement, the service RRH set is $ \Phi_{l_2}=\{2,3\} $.

\begin{figure}[!t]
	\centering
	{	\includegraphics[width=\hsize]{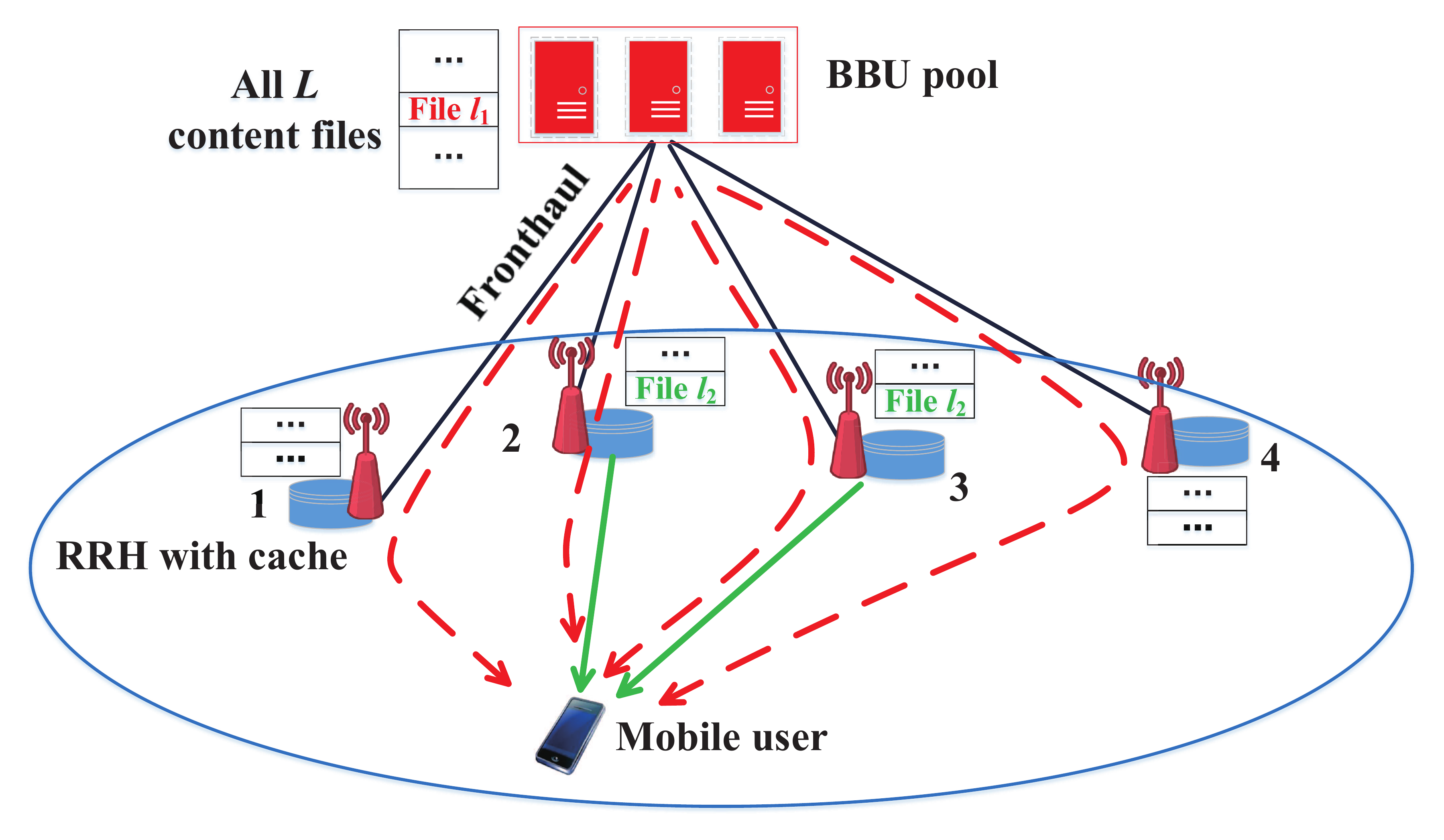}}
	\caption{System model and file delivery scheme. Red dashed and green solid lines represent the file fetching routes when user requests for the $ l_1 $-th and $ l_2 $-th content file, respectively.}
	\label{system model}
\end{figure}

The wireless channel is assumed to be block-fading, i.e., the channel's gain is kept as constant within the duration of a block, and different blocks experience independent and identically distributed (i.i.d.) fading. When being requested, a file would be transmitted through different blocks of the wireless channel. Assuming that both the RRH and the user's device are equipped with single antenna, the user's received signal from the service RRH set when requesting for the $ l $-th file can be expressed as
\begin{equation}
\label{channel} y = \sum\limits_{n\in\Phi_l}^{}\sqrt{p_T K d_n^{-\alpha}} h_n s+noise,
\end{equation}
where $ p_T $ is the transmit power of each RRH, $ K $ is a constant depending on the antenna characteristics and the average channel attenuation, $ d_n $ is the distance between the $ n $-th RRH and the user, $ \alpha $ is the path loss exponent, $ h_n \sim \mathcal{CN}(0,1)  $ represents complex Gaussian small scale fading, $ s $ represents the transmitted symbol and $ \mathbb{E}\left[ |s|^2 \right] =1 $, and $ noise $  denotes  complex additive white Gaussian noise (AWGN) with zero mean and variance $ \sigma^2 $.

The main modeling parameters and notations are summarized in Table \ref{model}.

\begin {table}[!h] 
\caption {Modeling Parameters and Notations} \label{model} 
\vspace{-1em}
\begin{center}
\begin{tabular}{ l  l  }
\hline
\textbf{Symbol}  &  ~~~ \textbf{Definition} \\ \hline
		$ N $ &~~~ Number of RRHs    \\
        $ L $ &~~~ Number of content files    \\
$ P_l $ &~~~ Request probability of the $ l $-th file    \\
$ \beta $ &~~~ Skewness factor of the Zipf's distribution    \\
$ M_n $ &~~~ The number of files that the $ n $-th RRH can cache    \\
$ \mathbf{A}^{L\times N} $ &~~~ Caching placement matrix    \\
$ a_{l,n} $ &~~~ Binary variable indicating whether the $ l $-th file is \\
             &~~~ cached in the $ n $-th RRH, $ (l,n) $-th entry of $ \mathbf{A}^{L\times N} $    \\
$ \Phi_l $ &~~~  The service RRH set for the user w.r.t. the $ l $-th file   \\
$ p_T $ &~~~ Transmit power of each RRH    \\
$ d_n $ &~~~ Distance between the $ n $-th RRH and the user    \\
		\hline
	\end{tabular}
\end{center}
\end{table}

\section{Problem Formulation and Analysis} \label{problem_formulation}

\subsection{Problem Formulation}

Define the normalized fronthaul usage w.r.t. the $ l $-th file as 
\begin{equation}
\label{Tl}
T_l(\mathbf{A})=\prod_{n=1}^{N}\left(1-a_{l,n}\right)=\left\{ 
\begin{array}{cl}
1,& \text{$ a_{l,n}=0$ for $\forall n $}\\
0,& \text{$ \exists n $ such that $ a_{l,n}=1 $}\\	
\end{array}
\right.,
\end{equation}
which indicates that if there is at least one copy of the requested file cached in the RRHs, there will be no fronthaul usage, i.e., $ T_l=0 $, while if the requested file is not cached in any of the RRHs, there will be fronthaul usage, i.e., $ T_l=1 $. Note that $ T_l $ does not depend on the user's location.

The caching strategy should be designed according to the long-term statistics over the user's locations and content file requests. The joint optimization problem can be formulated through a weighted sum of the objectives \cite{multicriteria},
\begin{subequations}
	\label{problem}
	\begin{eqnarray}
	\label{p1}      &\hspace{-0.3cm} \min &\hspace{-0.3cm} f_{obj}(\mathbf{A})=\eta\underbrace{\sum\limits_{l=1}^{L}P_l \mathbb{E}_{x_0}\left[ {P}_{out}^{(l)}(x_0) \right]}_{\textnormal{cell average outage probability}} +(1-\eta)\underbrace{\sum\limits_{l=1}^{L}{P}_l T_l}_{\textnormal{fronthaul usage}},\nonumber\\
	\\
	\label{c1}   &\hspace{-0.3cm} \textit{s.t.}&\hspace{-0.3cm} \sum\limits_{l=1}^{L}a_{l,n} = M_n, \\
	\label{c2}     & \hspace{-0.3cm}          &\hspace{-0.3cm} a_{l,n}\in\{0,1\}.
	\end{eqnarray}
\end{subequations}
where $ \eta \in [0,1] $ is a weighting factor to balance the tradeoff between outage probability and fronthaul usage, $ \mathbb{E}_{x_0} $ denotes expectation in terms of the user's location $ x_0 $, $ P_{out}^{(l)}(x_0) $ is the outage probability when the user requests for the $ l $-th file at location $ x_0 $. Constraint (\ref{c1}) describes the caching limit of each RRH, and constraint (\ref{c2}) indicates the joint optimization as a 0-1 integer problem.

Different values of $ \eta $ will lead to different balances between outage probability and fronthaul usage. Given $ \eta $, the caching strategy can be determined through solving the optimization problem in (\ref{problem}). 
In practice, $ \eta $ is chosen by the decision maker (e.g., RAN's operator) according to the system's long-term statistics of outage probability and fronthaul usage. For example, when the fronthauls' average payload is heavy, a small value of $ \eta $ should be chosen to reduce the fronthaul usage, and the price is to increase the outage probability. On the other hand, when the cell average outage probability is high, a large value of $ \eta $ should be chosen to reduce the outage probability, so that the price is to increase the fronthaul usage.

\subsection{Outage Probability Analysis}
When the user requests for the $ l $-th file at location $ x_0 $, the SNR of the received signal is given by
\begin{equation}
\label{gamma_n}
\gamma_l(x_0) = \sum_{n \in \Phi_l}^{}\frac{p_T}{\sigma^2}K d_n^{-\alpha}|h_n|^2
=\sum_{n \in \Phi_l}^{}\gamma_0 S_n|h_n|^2=\sum_{n \in \Phi_l}^{}\gamma_n,
\end{equation}
where $ \gamma_0=\frac{p_T}{\sigma^2} $ 
is SNR at the transmitter of each RRH, $ S_n=K d_n^{-\alpha} $ is the large scale fading, and $ \gamma_n=\gamma_0 S_n|h_n|^2 $ represents the received SNR from the $ n $-th RRH. For a specific file, without ambiguity, we omit the subscript of file index $ l $ and the user's location $ x_0 $ in the following analysis.

In the service RRH set $ \Phi $ with cardinality $ |\Phi| $, the RRHs with the same distance to the user are grouped together. Assuming there are $ I ~(I\leq |\Phi |) $ groups, the number of RRHs in the $ i $-th group is denoted by $ J_i $, and $ \sum_{i=1}^{I} J_i=|\Phi| $. The distance between the user and the RRH in the $ i $-th group is denoted by $ d_i~(i\in \{1,2,3,\cdots,I\}) $. 
Letting $ \lambda_i = \frac{1}{\gamma_0 K d_i^{-\alpha}} $, the probability density function (PDF) of the received SNR can be obtained as
\begin{equation}
\label{the_pdf_snr}
f_{\gamma}(\gamma)=\sum_{i=1}^{I}\sum_{j=1}^{J_i}\dfrac{\lambda_i^j A_{ij}}{(j-1)!}\gamma^{j-1}e^{-\lambda_i\gamma},
\end{equation}
and the cumulative distribution function (CDF) is given by
\begin{equation}
\label{the_cdf_snr}
\begin{aligned}
F_{\gamma}(\gamma)&=\sum_{i=1}^I\sum_{j=1}^{J_i}\dfrac{\lambda_i^{j-1}A_{i j}}{(j-1)!}\\
&~\cdot\left[\dfrac{(j-1)!}{\lambda_i^{j-1}}-\left(e^{-\lambda_i\gamma}\sum_{k=0}^{j-1}\dfrac{(j-1)!}{(j-1-k)!\lambda_i^k}\gamma^{j-1-k}\right)\right],
\end{aligned}
\end{equation}
where
\begin{equation}
\label{the_A_nm}
\begin{aligned}
&A_{ij}=\dfrac{(-\lambda_i)^{J_i-j}}{(J_i-j)!}\dfrac{d^{J_i-j}}{ds^{J_i-j}}\left[ M_{\gamma}(s)\left(1-\dfrac{1}{\lambda_i}\cdot s\right)^{J_i} \right]\Bigg|_{s=\lambda_i}
\end{aligned},
\end{equation}
and
\begin{equation}
M_\gamma (s)=\prod\nolimits_{n\in \Phi}\frac{1}{1-\gamma _0 S_n \cdot s}\;\;.
\end{equation}
The derivations of (\ref{the_pdf_snr}) and (\ref{the_cdf_snr}) are given in Appendix \ref{A1}.

When the distance between any service RRH and the user is distinct, i.e., $ d_n\neq d_m, \forall n\neq m\in \Phi, $ (\ref{the_pdf_snr}) and (\ref{the_cdf_snr}) are written as
\begin{equation}
\label{PDF_gamma}
\begin{aligned}
f_{\gamma} (\gamma)&=&\sum_{n \in \Phi}^{}\dfrac{1}{\gamma_0 S_n} \left( \prod_{\substack{m \in \Phi\\m \neq n}}^{}\dfrac{S_n}{S_n-S_m}\right) \exp\left( {-\dfrac{\gamma}{\gamma_0 S_n}} \right)
\end{aligned}
\end{equation}
and 
\begin{equation}
\label{CDF_gamma}
\begin{aligned}
F_\gamma (\gamma)=\sum_{n \in \Phi}^{} \left( \prod_{\substack{m \in \Phi\\m \neq n}}^{}\dfrac{S_n}{S_n-S_m}\right) \left[ 1-\exp\left( {-\dfrac{\gamma}{\gamma_0 S_n}} \right) \right],
\end{aligned}
\end{equation}
respectively.

The accuracy of the derived CDF of (\ref{the_cdf_snr}) (written as (\ref{CDF_gamma}) in special case) is illustrated in Fig. \ref{result_0_CDF} through three scenarios. Assuming there are 6 service RRHs for the user, and the distances between the service RRHs and the user are denoted by a vector $ \mathbf{D} $. The three different scenarios are (1) scenario 1: $ \mathbf{D_1}=[0.8R,0.8R,0.8R,0.8R,0.8R,0.8R] $, ($ R $ is the cell radius), i.e., all the RRHs are with the same distance to the user; (2) scenario 2: $ \mathbf{D_2}= [0.6R,0.7R,0.7R,0.8R,0.8R,0.8R] $, i.e., some of the RRHs have same distance with the user; (3) scenario 3: $\mathbf{D_3}=[0.5R,0.6R,0.7R,0.8R,0.9R,1.0R] $, i.e., all the RRHs are with different distances to the user. It can be seen from Fig. \ref{result_0_CDF} that the analytical results match the simulation results, which demonstrates the accuracy of the derived expression of (\ref{the_cdf_snr}).

\begin{figure}[!t]
	\centering
	{	\includegraphics[width=\hsize]{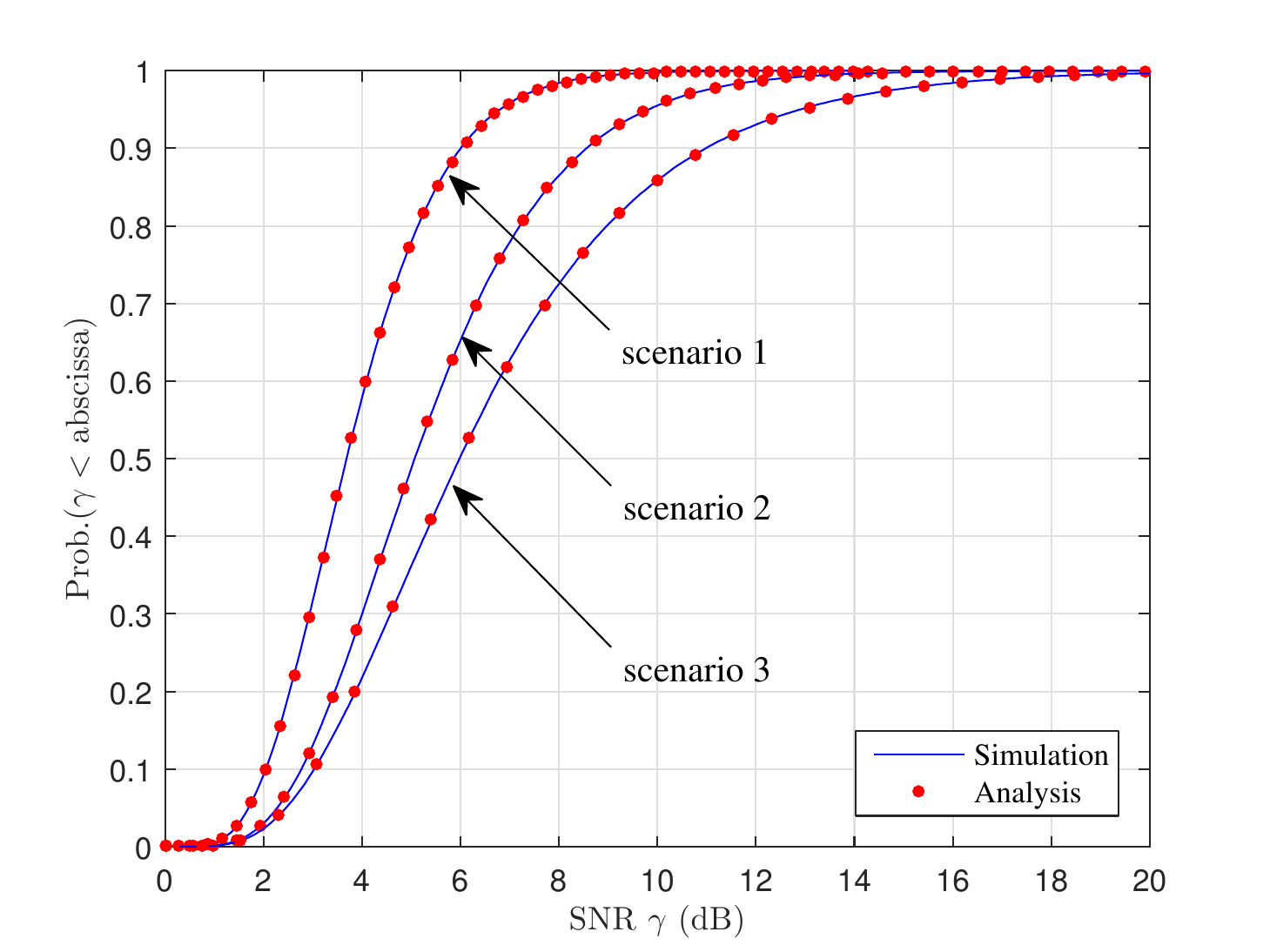}}
	\caption{CDF of the user's received SNR at a fixed location.}
	\label{result_0_CDF}
\end{figure}

The outage probability according to a certain SNR threshold $ \gamma_{th} $ is 
\begin{equation}
\label{out_2_CDF}
P_{out}(\gamma_{th})=F_{\gamma}(\gamma_{th}).
\end{equation}

It is difficult to find a closed form solution of the cell average outage probability w.r.t. the $ l $-th file, i.e., $ \mathbb{E}_{x_0}[P_{out}^{(l)}(x_0)] $. However, we can use the composite Simpson's integration in forms of polar coordinates, where the user's location is denoted by $ (\rho,\theta) $ and $ x_0=\rho e^{j\theta} $.
\begin{equation}
\label{ex}
\begin{aligned}
&\mathbb{E}_{x_0}\left[ {P}^{(l)}_{out}(x_0) \right]  \\
=&\int_{0}^{2\pi}\int_{0}^{R} {P}^{(l)}_{out}(\rho,\theta)f_{x_0}(\rho,\theta)\rho d \rho d \theta  \\
\approx & \dfrac{\Delta h \Delta k}{9}\sum_{u=0}^{U}\sum_{v=0}^{V}w_{u,v}\rho_u {P}^{(l)}_{out}(\rho_u,\theta_v)f_{x_0}(\rho_u,\theta_v), \\
\end{aligned}
\end{equation}
where $ R $ is the cell radius, even integers $ U $ and $ V$ are chosen such that $ \Delta h=R/U $ and $ \Delta k =2\pi/V $ meeting the requirement of calculation accuracy, $ \rho_u=u\Delta h$, $ \theta_v=v\Delta k$, $ f_{x_0}(\rho,\theta)$ is the probability density function of the user's location, which is $ 1/\pi R^2 $ when the user's location is uniformly distributed in the cell, and $ \{w_{u,v}\} $ are constant coefficients (please refer to \cite{system_outage} and Chapter 4 in \cite{book}).

Substituting (\ref{Tl}), (\ref{the_cdf_snr}), (\ref{out_2_CDF}) and (\ref{ex}) into (\ref{p1}), the optimization problem is formulated as a  function of the caching placement matrix $ \mathbf{A}^{L\times N}=\{a_{l,n}\} $. However, the problem is a 0-1 integer nonlinear problem, and it is difficult to obtain a closed form solution. The following section will focus on how to solve this problem.

\section{Caching Placement Scheme} \label{caching_placement}

In this section, two efficient approaches are proposed to solve the joint optimization problem: one is GA-based approach and the other is mode selection approach.

\subsection{Genetic Algorithm Based Approach}

Genetic algorithm is inherently suitable for solving optimization problems with binary variables \cite{ga_survey}. The algorithm structure is shown in Fig. \ref{ga_structure}. Firstly, $ N_p $ candidate caching placement matrices are generated, known as the initial population (with population size $ N_p $), and each matrix is called an individual. Then the objective value of each individual is evaluated through (\ref{p1}). $ N_e $ individuals with best objective values are chosen as elites and passed into next generation (children of current generation population) directly. The rest of the next generation population are generated through crossover and mutation operations. The crossover function operates on two individuals (known as parents) and generates a crossover child, and the mutation function operates on a single individual and generates a mutation child. The number of individuals generated through crossover and mutation operations are denoted as $ N_c $ and $ N_m $, respectively, where  $ N_e+N_c+N_m=N_p $, and the crossover fraction is defined as $ f_c=\frac{N_c}{N_c+N_m} $. The selection function selects $ 2N_c $ and $ N_m $ individuals from the current generation for the crossover and mutation function, respectively, where some individuals will be selected more than once. Stochastic uniform sampling selection  \cite{introduction_GA} is adopted, and individuals with lower objective values in current generations will have a higher probability to generate offsprings. Repeat the evaluation-selection-generation procedures until termination criterion is reached. Finally, the best individual in the current population is chosen as the output of the algorithm. The initial population, crossover function and mutation function of the proposed GA approach are described as follows.

\subsubsection{Initial Population}

The initial population is created as a set of $ \{\mathbf{A}^{L\times N}\} $. For each column in each individual,  $ M_n $ out of the first $ L' $ entries (i.e., $ \{a_{1,n}, a_{2,n}, \cdots, a_{L',n}\} $) are set to be one randomly, and all the remaining entries  are set to be zero, where
\begin{equation}
L' = \sum\nolimits_{n=1}^{N} M_n < L
\end{equation}
is based on the fact that the total different files with higher popularity can be cached in the RRHs are $ \{F_l|l=1,2,\cdots,L'\}  $. There is no benefit to cache files  $ \{F_l|l>L'\} $ with lower popularity.

\subsubsection{Crossover Function}

The crossover function generates a child $ \mathbf{A}_c $ from parents $ \mathbf{A}_1 $ and $ \mathbf{A}_2 $. A two-point crossover function is used, which is described in Algorithm \ref{crossover}, in which steps \ref{s1} to \ref{s4} are heuristic operations to meet constraint (\ref{c1}).

\IncMargin{0em} 
\begin {algorithm}[!h]
\caption {Crossover function}
\label{crossover}
\SetAlgoLined
Get parent $ \mathbf{A}_1=\{a_{l,n}^{(1)}\} $ and $ \mathbf{A}_2=\{a_{l,n}^{(2)}\} $ from selection function, initialize their child $ \mathbf{A}_c=\{a_{l,n}^{(c)}\}=\mathbf{0}^{L\times N} $.\\
\For {$ n=1,2,\cdots,N $} {
	Generate random integers $ l_1,l_2 \in [1,L'] $, $ l_1\neq l_2 $ \\
	\eIf{$ l_1<l_2 $}{Replace $ a_{l,n}^{(1)},~l=\{l_1,l_1+1,\cdots,l_2\} $ of $ \mathbf{A}_1 $ with $ a_{l,n}^{(2)},~l=\{l_1,l_1+1,\cdots,l_2\} $ of  $ \mathbf{A}_2 $, and then set $ a^{(c)}_{l,n}=a^{(1)}_{l,n}, ~\forall l \in\{1,2,\cdots,L\} $.\\}
	{ Replace $ a_{l,n}^{(2)},~l=\{l_2,l_2+1,\cdots,l_1\} $ of $ \mathbf{A}_2 $ with $ a_{l,n}^{(1)},~l=\{l_2,l_2+1,\cdots,l_1\} $ of  $ \mathbf{A}_1 $, and then set $ a^{(c)}_{l,n}=a^{(2)}_{l,n}, ~\forall l \in\{1,2,\cdots,L\} $.\\} 
	\While{$ \sum\nolimits_{l=1}^{L}a^{(c)}_{l,n} > M_n $ 	\label{s1} }{
		Set nonzero $ a^{(c)}_{l,n} $ to 0 in descending order of $ l $.
	}
	\label{s2}
	\While{$ \sum\nolimits_{l=1}^{L}a^{(c)}_{l,n} < M_n $ \label{s3}}{
		Set zero $ a^{(c)}_{l,n} $ to 1 in ascending order of $ l $.         
	}
	\label{s4}}	
\end{algorithm}
\DecMargin{0em}

\subsubsection{Mutation Function}

The mutation function operates on a single individual and generates its mutation child. For each column of the individual, one of the first $ L' $ entries is randomly selected and the value is set to be the opposite ($ 0 $ to $ 1 $ and vice versa), then steps \ref{s1} to \ref{s4} described in Algorithm \ref{crossover} are executed to meet constraint (\ref{c1}). The mutation operation reduces the probability that the algorithm converges to local minimums.

If $ N_g $ generations are evaluated, there is a total of $ N_pN_g $ calculations of the objective values. In order to further reduce the computational complexity of caching strategy, a mode selection approach is proposed in next subsection.

\begin{figure*}[!t]
	\centering
	{	\includegraphics[width=\hsize]{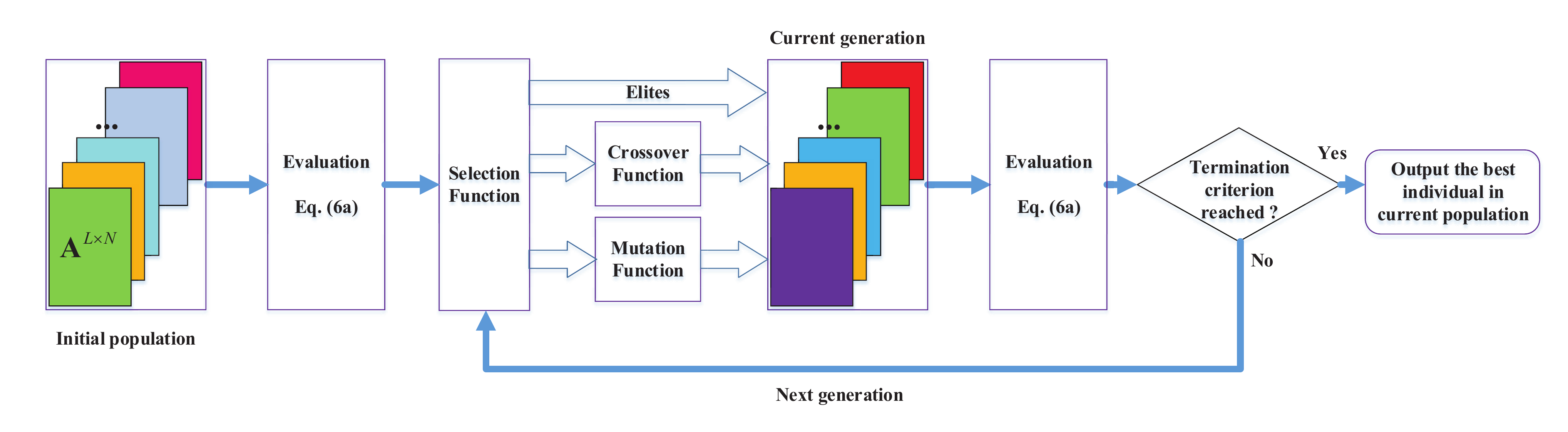}}
	\caption{Genetic algorithm structure.}
	\label{ga_structure}
\end{figure*}

\subsection{Mode Selection Approach} \label{two-mode}

There are two particular caching placement schemes: one is the most popular content (MPC) caching, and the other one is the largest content diversity (LCD) caching \cite{backhaul-aware, MPC_original, LCD_original}. 
In MPC, each RRH caches the most popular files, i.e., the $ n $-th RRH caches $ \{F_l|l=1,2,\cdots,M_n\} $, which will have low outage probability while high fronthaul usage. In the LCD scheme, a total of $ L'=\sum\nolimits_{n=1}^{N} M_n ~(< L) $ different most popular content files are cached in the RRHs, which can have lowest fronthaul usage while relatively high outage probability. If the LCD scheme is adopted in Cloud-RAN, the impact of locations of caching content files on the cell average outage probability needs to be considered. Assuming the locations of the user are uniformly distributed in the cell, caching the most popular files in the RRH nearest to the cell center will achieve better outage probability performance, which is similar to the RRH placement problem \cite{antenna}. Therefore, for Cloud-RAN, we improve the LCD scheme and propose a location-based LCD (LB-LCD) scheme which is described in Algorithm \ref{LCD}.

\IncMargin{0em} 
\begin {algorithm}[!h]
\caption {Proposed LB-LCD caching strategy}
\label{LCD}
\SetAlgoLined
Sort the RRH set as $ \mathcal{N}_s = \{n_i|i=1,2,\cdots,N, \; \;  D_{n_1}\leq D_{n_2}\leq \cdots \leq D_{n_N}\}$, where $ D_{n_i} $ denotes the distance between the $ n_i $-th RRH and the cell center. \\
Fill the cache of the RRH set  $ \mathcal{N}_s $ in sequence from $ n_1 $ to $ n_N $ with content files $ \{F_l|l=1,2,\cdots,\sum\nolimits_{n=1}^{N} M_n\} $ in ascending order of $ l $.\\
\end{algorithm}
\DecMargin{0em}

For example, there are 3 RRHs $ \{1,2,3\} $, and each RRH can cache 3 files, so that all the RRHs can cache 9 different content files. The distance between RRH $ i $ and the cell center is $ D_i $, assuming $ D_1<D_2<D_3 $. The LB-LCD caching strategy is illustrated in Table \ref{LCD_example}.

\begin {table}[!h] 
\caption {Example of The LB-LCD Caching Strategy} \label{LCD_example} 
\vspace{-1em}
\begin{center}
	\begin{tabular}{ c  c  c c}
		\hline
		\textbf{RRH} & \textbf{1} &  \textbf{2}  &   \textbf{3}\\ 
		\hline
		\begin{tabular}{c} Files \\cached \end{tabular} & \begin{tabular}{c} $ F_1 $\\$ F_2 $\\$ F_3 $		\end{tabular} & \begin{tabular}{c} $ F_4 $\\$ F_5 $\\$ F_6 $\end{tabular}	& \begin{tabular}{c} $ F_7 $\\$ F_8 $\\$ F_9 $ \end{tabular}	\\
		\hline
	\end{tabular}
\end{center}
\end{table}

\begin{prop}
	\label{prop1}
	The objective value of both the MPC scheme and the LCD scheme is linear with $ 
	\eta $. When $ \eta $ is small, i.e., minimization of the fronthaul usage is weighted more, the LCD scheme is superior to the MPC scheme. When $ \eta $ is large, i.e., minimization of the cell average outage probability is weighted higher, the MPC scheme is superior to the LCD scheme. There exists a crossover point of the two schemes, the weighting factor of the crossover point is
	\begin{equation}
	\label{eta_caculation}
	\begin{array}{c}
	\eta_0=\dfrac{1}{1+\dfrac{\sum\nolimits_{l=1}^{L}P_l\mathbb{E}_{x_0}\left[ P_{out,MPC}^{(l)}(x_0)-P_{out,LCD}^{(l)}(x_0)\right] }{\sum\nolimits_{l=1}^{L}P_l\left( T_{l,LCD}-T_{l,MPC}\right)}}
	\end{array}.
	\end{equation}
	When $ M_n = M,~\forall n $, $ \eta_0 $ can be further expressed as
	\begin{equation}
	\label{eta_caculation_2}
	\begin{array}{c}
	\eta_0=\dfrac{1}{1+\dfrac{\sum\nolimits_{l=1}^{NM}P_l \mathbb{E}_{x_0}\left[P_{out,LCD}^{(l)}(x_0)-P_{out,MPC}^{(l)}(x_0)\right] }{\sum\nolimits_{l=M+1}^{NM}P_l}}
	\end{array}.
	\end{equation}
\end{prop}

\begin{proof}
Please refer to Appendix \ref{A3}.
\end{proof}

Based on proposition \ref{prop1}, we propose a 
mode selection caching strategy. The RAN can make a decision of the tradeoff according to the statistics of cell average outage probability and fronthaul usage in the cell, and a tradeoff weighting factor $ \eta $ is chosen. When $ \eta \leq \eta_0 $, select the LB-LCD caching scheme, while when $ \eta > \eta_0 $, select the MPC caching scheme.

\subsection{Computational Complexity Analysis}

The number of objective function calculations w.r.t. a certain value of $ \eta $ is evaluated to measure the complexities of the exhaustive search method, the proposed GA approach and the proposed mode selection approach. The complexity of exhaustive search is $ \prod_{n=1}^{N} \binom{L}{M_n} $. When $ M_n=M,~ \forall n $, it is clear that the complexity of exhaustive search is exponential w.r.t. the number of RRHs, i.e., $ \binom{L}{M}^N $. The complexity of the proposed GA is $ N_pN_g $, where $ N_p $ and $ N_g $ are the population size and the number of generations evaluated, respectively.  $ N_g $ is determined by the convergence behavior of the GA. While the complexity of the proposed mode selection scheme is only 2. The reason is that, once the value of $ \eta_0 $ is solved from  (\ref{eta_caculation}), the RAN can choose a mode between MPC and LCD based on whether $ \eta>\eta_0 $, and 2 objective function calculations are involved in solving the equation. Further more, once $ \eta_0 $ is obtained,  caching schemes for all values of $ \eta $ are obtained. The computational complexities of the three approaches are summarized in Table \ref{complexity}.

\begin {table}[!h] 
\caption {Computational Complexity} \label{complexity} 
\vspace{-1em}
\begin{center}
	\begin{tabular}{ l  l  }
		\hline
		\textbf{Scheme}  &  \textbf{Objective function calculations} \\ \hline
	Exhaustive search &~~~~~ $ \prod\nolimits_{n=1}^{N}\binom{L}{M_n}$    \\
		Proposed GA-based approach &~~~~~ $ N_pN_g $ \\
		Proposed mode selection approach &~~~~~ 2 \\
		\hline
	\end{tabular}
\end{center}
\end{table}

\section{Numerical Results} \label{numerical_results}

In this section, the performances of the proposed two caching strategies are investigated through some representative numerical results. Firstly, the accuracy of the cell average outage probability and fronthaul usage analysis are verified by evaluating two typical caching schemes, i.e., the MPC and the LB-LCD schemes. Then the effectiveness of the proposed GA approach is verified by comparing its performance with exhaustive search, where the Pareto optimal solutions \cite{multicriteria} of the joint optimization problem are presented. In the proposed GA approach, placement matrices of the MPC and the LB-LCD schemes are added into the initial population to further improve the performance. Finally, performances of different caching strategies are compared and the convergence behavior of the proposed GA is presented. 

The MATLAB software is used for the Monte-Carlo simulations and numerical calculations. Throughout the simulation, it is assumed that each RRH has the same cache size $ M_n = M $. The transmit power of each RRH is $ p_T = \frac{P}{N} $,  where $ P $ is the total transmit power in the cell and $ \frac{P}{\sigma^2}=23 ~\text{dB} $. The constant $ K $ in (\ref{channel}) is chosen such that the received power attenuates 20 dB when the distance between the RRH and the user is $ R $ \cite{transmission_schemes}. In such setting, the outage probability does not depend on the absolute value of $ R $, that is, $ R $ can be regarded as the normalized radius. The main simulation parameters are summarized in Table \ref{parameter}.

\begin {table}[!h] 
\caption {Simulation Parameters} \label{parameter} 
\vspace{-1em}
\begin{center}
	\begin{tabular}{ l  c  }
		\hline
		\textbf{Parameter}  & \textbf{Value} \\ \hline
		Path loss exponent $ \alpha $ 	& $ 3 $\\
		$ P/\sigma^2 $ & $23$ dB \\
		SNR threshold 	 $ \gamma_{th} $ &  $ 3 $ dB\\
		User location distribution & uniform\\
		$ U $ and $ V $ in Simpson's integration & $ 6, 6 $\\
		Population size $ N_p $ in GA & $ 50 $ \\	
		Selection function  & stochastic universal sampling\\
		Number of elites $ N_e $ 	& $ 10 $ \\
		Crossover fraction $ f_c $ & $ 0.85 $ \\
		\hline
	\end{tabular}
\end{center}
\end{table}

\subsection{MPC and LB-LCD Caching Placements}

Cell average outage probability ($ \sum\nolimits_{l=1}^{L}P_l \mathbb{E}_{x_0}[ P_{out}^{(l)}(x_0) ] $) and average fronthaul usage ($ \sum\nolimits_{l=1}^{L}P_l T_l $) of the MPC and LB-LCD schemes are shown in Fig. \ref{result_1_a} and Fig. \ref{result_1_b}, respectively. There are $ L=50 $ files, $ N=7 $ RRHs with one RRH located at the cell center and the other 6 RRHs evenly distributed on the circle with radius $ {2R}/{3} $ \cite{delivery_content_journal, DAS_randomness}, and each RRH can cache $ M=5 $ files. Both the simulation and numerical results are shown in this subsection. In the Monte-Carlo simulations, there are $ 10^4 $ realizations of the user's different locations and content file requests.

Cell average outage probability with different SNR threshold $ \gamma_{th} $ and popularity skewness factor $ \beta $ is illustrated in Fig. \ref{result_1_a}. It can be seen that the outage probability of both the MPC and the LB-LCD schemes increases with the increase of $ \gamma_{th} $, and the outage probability of the MPC scheme is lower than that of the LB-LCD scheme. The MPC curves of different values of $ \beta $ coincide. The reason is as follows: according to the file delivery scheme and the MPC caching strategy, no matter whether the requested file is cached in the RRHs or not, the file will be transmitted to the user from all the RRHs, thus the cell average outage probability w.r.t. any $ l $-th file is the same, denoting as  $ \mathbb{E}_{x_0}[ P_{out}^{(l)}(x_0) ]=P_{cell,out}$, then the  cell average outage probability expected on all the file requests is
\begin{equation}
\sum\limits_{l=1}^{L}P_l \mathbb{E}_{x_0}[ P_{out}^{(l)}(x_0) ] = P_{cell,out}\cdot\sum\limits_{l=1}^{L}P_l= P_{cell,out} , 
\end{equation}
which is not related to $\beta $, i.e., no matter how the popularity is distributed over the files, the cell average outage probability is kept as a constant.

\begin{figure}[!t]
	\centering
	{\includegraphics[width=\hsize]{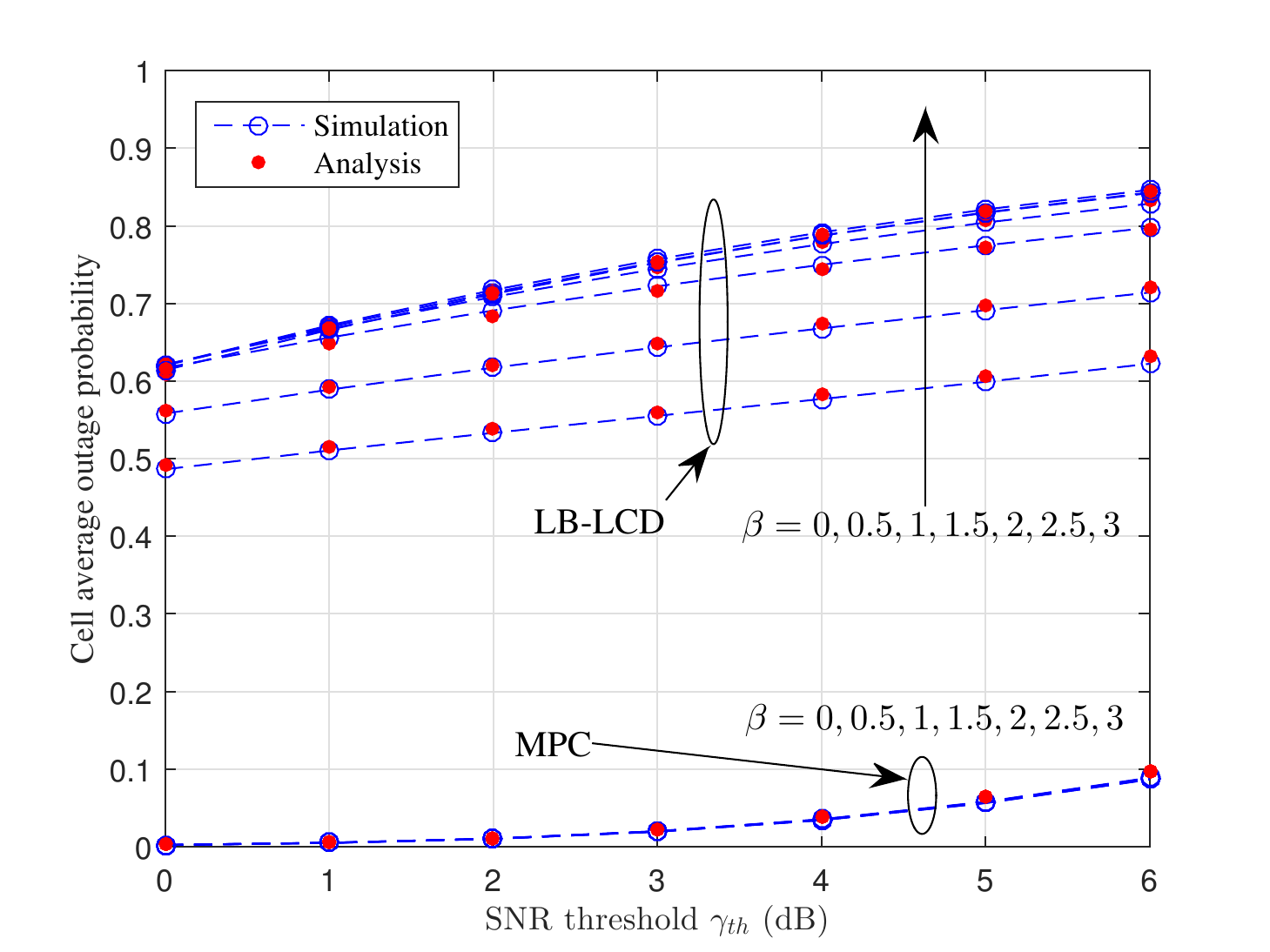}} 
	\caption{Cell average outage probability. $ L = 50, M=5, N=7$.}
	\label{result_1_a}
\end{figure}

For the LB-LCD scheme, cell average outage probability reaches the minimum value when $ \beta=0 $,  and it increases as $ \beta $ increases and approaches the maximum value when $ \beta $ is large enough, e.g., $ \beta = 2,2.5,3 $. The reason is explained as follows. According to the file delivery scheme and the LB-LCD caching strategy, if the requested file is not cached in any of the RRHs, the file will be fetched from the BBU pool, and then transmitted to the user from all the RRHs. The outage probability will then achieve the minimum value due to wireless diversity. While if the requested file is cached in the RRH (only cached in one of the RRHs), the file will be transmitted to the user from only one RRH, and the outage probability will be relatively higher. When $ \beta=0 $,  $ P_l=1/L, \forall l $, i.e., the request probability is the same for all the content files, which means that the cell average outage probability depends evenly on the outage probability of each file, and the outage probability of the files which are not cached in the RRHs is lower than that of the files cached in the RRHs. As $ \beta $ increases, the more skewness of the popularity will toward the first few files with high ranks, i.e., the cell average outage probability depends more on these files, and there is a higher probability that there is only one copy for each of these files cached in one certain RRH, and the corresponding outage probability is high, so the outage probability increases as $ \beta $ increases and the curve with $ \beta=0 $ is the lower bound. Note that  
\begin{equation}
\label{sum1_5}
 \sum_{l=1}^{5}P_l = \left\lbrace
 	\begin{aligned}
 	0.90, & ~\beta=2.0\\
 	0.96, & ~\beta=2.5\\
 	0.99, & ~\beta=3.0
 	\end{aligned}\right.
\end{equation}
which means when $ \beta $ is large enough ($ \beta\geqslant 2.0 $), the cell average outage probability depends mainly on the first 5 most popular files. These 5 files are cached in the RRH located at the cell center, and the cell average outage probability w.r.t any one of the 5 files is the same, so the cell average outage probability is nearly the same for different values of $ \beta ~(\geqslant 2.0) $, which approaches the maximum value.

Fig. \ref{result_1_b} shows the fronthaul usage of the two caching schemes. Because the fronthaul usage is independent of $ \gamma_{th} $, the curve versus different values of the skewness factor $ \beta $ is evaluated. The LB-LCD scheme has lower fronthaul usage than the MPC scheme, which is because that the LB-LCD scheme caches a total of $ MN=5\times7=35 $ different files in the RRHs while the MPC scheme caches only $ M=5 $ different files. The average fronthaul usage of both the MPC and LB-LCD scheme decreases with the increase of $ \beta $, which is due to the same reason that as $ \beta $ increases, the popularity becomes more skewed towards the first few files with higher ranks, and there is a higher probability that these few files are cached in the RRHs. As shown in  (\ref{sum1_5}), when $ \beta=3 $, the fronthaul usage almost depends on the first 5 popular files, since they are all cached in the RRHs under both the MPC and the LB-LCD caching strategies, the fronthaul usages of both schemes approach zero.

It is seen from Fig. \ref{result_1_a} and Fig. \ref{result_1_b} that the analytical results are highly consistent with the simulation results. Therefore, analytical results will be used instead of time-consuming simulations in the following evaluations.

\begin{figure}[!t]
	\centering
	{\includegraphics[width=\hsize]{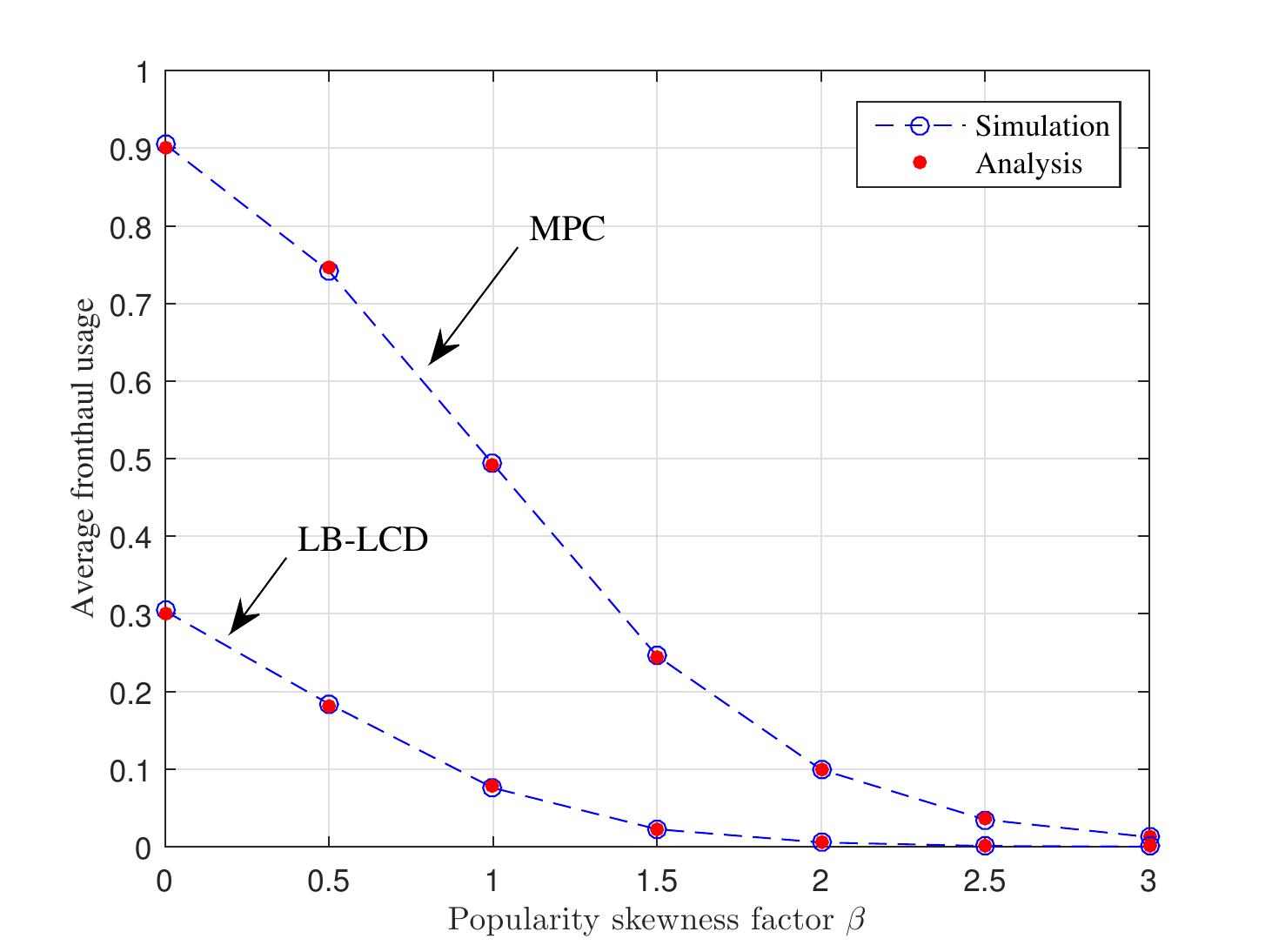}} 
	\caption{Cell average fronthaul usage. $ L = 50, M=5, N=7 $.}
	\label{result_1_b}
\end{figure}

\subsection{Tradeoffs between Cell Average Outage Probability and Fronthaul Usage}

Tradeoffs between cell average outage probability and fronthaul usage obtained by exhaustive search and the proposed GA-based approach are shown in this subsection. There are three RRHs, and the polar coordinates of which are $ \left(\frac{R}{4},0\right) $, $ \left(\frac{R}{3},\frac{2\pi}{3}\right) $, and $ \left(\frac{R}{2},\frac{4\pi}{3}\right) $, respectively. There are $ L=9 $ content files, and the popularity skewness factor $ \beta = 1.5 $.

\begin{figure}[!t]
	\centering
	{\includegraphics[width=\hsize]{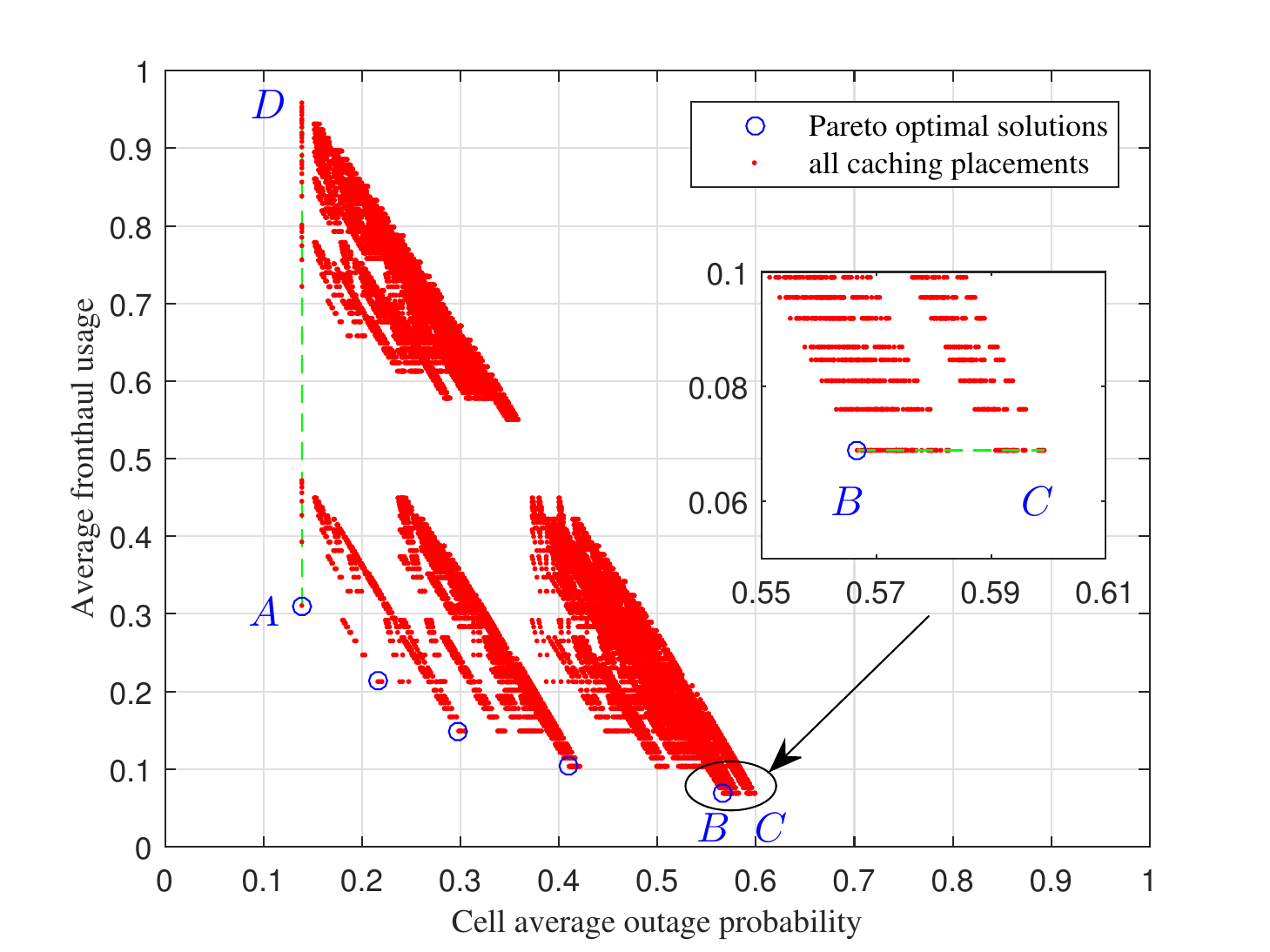}}
	\caption{Cell average outage probability and fronthaul usage tradeoff region. $ L = 9, M = 2, N = 3,  \beta = 1.5$. Each red point corresponds to a caching placement, and the 5 points emphasized by small blue circles are the Pareto optimal solutions of the joint optimization problem.}
	\label{pareto_front}
\end{figure}

Fig. \ref{pareto_front} is focused on the scenario when the cache size $ M $ is equal to $ 2 $. All tradeoffs between cell average outage probability and fronthaul usage are given by exhaustive search. Since the popularities of the content files $ \{ P_l \} $, the fronthaul usage $ \{T_l\} $ and the outage probability $ \{ P^{ (l) }_{out}\} $ are all discontinuous values w.r.t. integer $ l $, the cell average outage probability and average fronthaul usage region of all caching placements is a set of discrete points as shown in the figure, where each red point corresponds to a caching placement. The 5 points emphasized by small blue circles are the Pareto optimal solutions (nondominated set \cite{multicriteria}) of the joint optimization problem, i.e., there is no other point dominating with the Pareto optimal solutions in terms of both the cell average outage probability and fronthaul usage.

The cell average outage probability is minimized when the files cached in each RRH are the same, and the popularity of these cached files will have an impact on the average fronthaul usage. The corresponding points of these caching placements lie on the line segment $ AD $, i.e., line segment $ AD $ represents the lower bound of the cell average outage probability. The MPC scheme represented by point $ A $ achieves the minimum fronthaul usage among these caching placements. The reason is that the MPC scheme caches the most popular files which can reduce the fronthaul usage to a minimum value among these caching placements.
   
The fronthaul usage is minimized when all RRHs cache different files with higher popularity ranks, and the cache locations of these files will have an impact on the cell average outage probability. The corresponding points of these caching placements lie on the line segment $ BC $, i.e., line segment $ BC $ represents the lower bound of the average fronthaul usage. The LB-LCD scheme represented by point $ B $ achieves the minimum cell average outage probability among these caching placements. The reason is that the LB-LCD scheme caches the files with higher ranks in the RRHs near to the cell center, which has the minimum cell average outage probability among these caching placements.

\begin{figure}[!t]
	\centering
	{\includegraphics[width=\hsize]{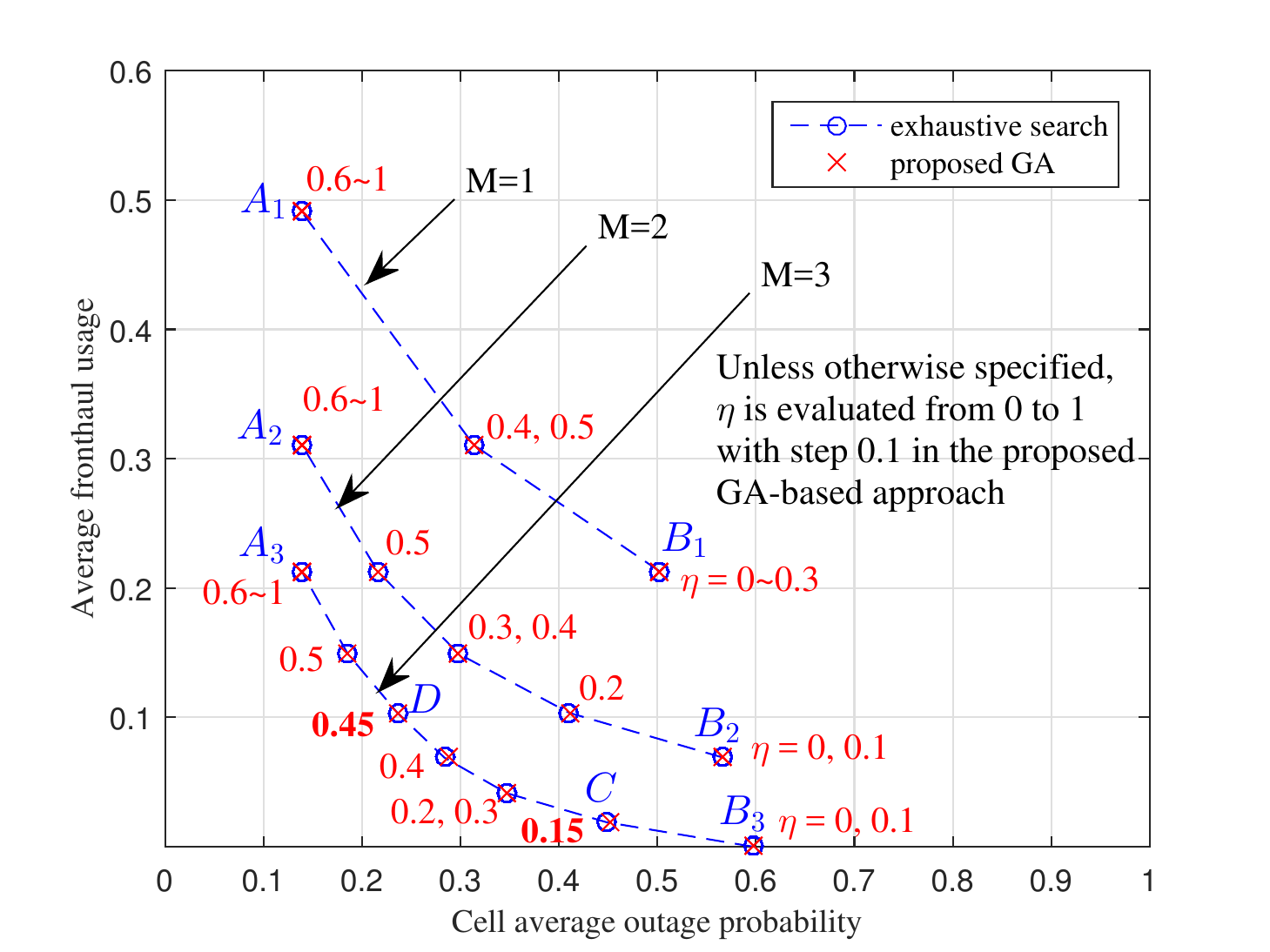}}
	\caption{Cell average outage probability and fronthaul usage tradeoff. $ L = 9, M = \{1, 2, 3\}, N = 3, \beta = 1.5$. Unless otherwise specified in the figure, $ \eta $ is evaluated from 0 to 1 with a step 0.1 in the proposed GA-based approach, i.e., $\eta = \{0.1,0.2,\cdots, 1\}$.}
	\label{ex_and_ga}
\end{figure}

Fig. \ref{ex_and_ga} shows the Pareto optimal tradeoffs between the cell average outage probability and the fronthaul usage with different cache size $ M $. The results obtained through the proposed GA approach are almost the same as exhaustive search, which means that the proposed GA approach can achieve near-optimal performance. The minimum cell average outage probability is achieved at point $ A_1 $ when $ M=1 $, $ A_2 $ when $ M=2 $, and $ A_3 $ when $ M=3 $, respectively. The minimum cell average outage probability represented by the three points are the same, and the corresponding caching placements of the three points are the MPC scheme. The reason is that according to the file delivery scheme and the MPC caching placement, all the RRHs will serve the user no matter how many files the RRHs can cache. It is also seen that the corresponding fronthaul usage of the three points decrease as $ M $ increases, which is obvious because larger cache size can cache more files thus the fronthaul usage can be reduced.

On the other hand, the minimum fronthaul usage is achieved at point $ B_1 $ when $ M=1 $, $ B_2 $ when $ M=2 $, and $ B_3 $ when $ M=3 $, respectively. The corresponding caching placements of the three points are the LB-LCD scheme. Obviously, the corresponding fronthaul usage of the three points decreases as $ M $ increases. The fronthaul usage is zero at point $ B_3 $ when $ M=3 $, the reason is that all the RRHs can cache a total of  $ MN=3\times 3 = 9 $ files, which is equal to the number of files in the file library, i.e., all the files are cached in the RRHs. The corresponding cell average outage probability of the three points increases as $ M $ increases. The reason is that according to the file delivery scheme and the LB-LCD caching strategy, more different files can be cached in the RRHs as $ M $ increases, however, there is only one copy of each file and the outage probability w.r.t. these cached files will be higher, i.e., more different files cached in the RRHs, higher the cell average outage probability is.

Note that as the cache size $ M $ increases, the GA-based approach should evaluate more values of $ \eta $ in order to obtain all the Pareto optimal solutions of the joint optimization problem. For example, when $ M=3 $, additional values of $ \eta = 0.15 $ and $ \eta = 0.45 $ are evaluated to obtain the Pareto optimal solutions represented by point $ C $ and $ D $.

Table \ref{placement_table} shows all the optimal caching placements obtained by the proposed GA-based approach 
when $ M=2 $. For illustration, we use a $ M\times N $ matrix to represent the caching placement, with the $ (m,n) $-th entry $ b_{m,n}\in\{1,2,3,\cdots,L\} $ denotes the file index cached in the $ m $-th cache space of the $ n $-th RRH. From (\ref{eta_caculation_2}), $ \eta_0 = 0.3312 $. It can be seen that the LB-LCD scheme is the optimal placement when $ \eta= 0,0.1<\eta_0 $, while the MPC scheme is the optimal solution when $\eta=0.6\sim 1.0>\eta_0 $, and some files are duplicately cached in the RRHs when $ \eta=0.2\sim 0.5 $. 

\begin {table}[!h] 
\caption {Optimal Caching Strategy Obtained by The Proposed GA} \label{placement_table} 
\vspace{-1em}
\begin{center}
	\begin{tabular}{ c | c | c }
		\hline
		$\begin{array}{c} \eta = 0 \\f_{obj}=0.0689\\  \left[\begin{array}{ccc} 1&3&5\\2&4&6 \end{array}\right] \end{array}$    &    $\begin{array}{c} \eta = 0.1 \\ f_{obj}=0.1186\\  \left[\begin{array}{ccc} 1&3&5\\2&4&6 \end{array}\right] \end{array}$     &     $\begin{array}{c} \eta = 0.2 \\ f_{obj}=0.1651\\  \left[\begin{array}{ccc} 1&1&4\\2&3&5 \end{array}\right] \end{array}$\\ \hline
		$\begin{array}{c} \eta = 0.3 \\ f_{obj}=0.1938\\  \left[\begin{array}{ccc} 1&1&1\\2&3&4 \end{array}\right] \end{array}$   &    $\begin{array}{c} \eta = 0.4 \\ f_{obj}=0.2087\\  \left[\begin{array}{ccc} 1&1&1\\2&3&4 \end{array}\right] \end{array}$    &    $\begin{array}{c} \eta = 0.5 \\ f_{obj}=0.2144\\  \left[\begin{array}{ccc} 1&1&1\\3&2&2 \end{array}\right] \end{array}$\\ \hline
		$\begin{array}{c} \eta = 0.6 \\ f_{obj}=0.2077\\  \left[\begin{array}{ccc} 1&1&1\\2&2&2 \end{array}\right] \end{array}$    &    $\begin{array}{c} \eta = 0.7 \\ f_{obj}=0.1905\\  \left[\begin{array}{ccc} 1&1&1\\2&2&2 \end{array}\right] \end{array}$    &    $\begin{array}{c} \eta = 0.8 \\ f_{obj}=0.1733\\  \left[\begin{array}{ccc} 1&1&1\\2&2&2 \end{array}\right] \end{array}$  \\ \hline
		$\begin{array}{c} \eta = 0.9 \\ f_{obj}=0.1561\\  \left[\begin{array}{ccc} 1&1&1\\2&2&2 \end{array}\right] \end{array}$   &    $\begin{array}{c} \eta = 1.0 \\ f_{obj}=0.1390\\  \left[\begin{array}{ccc} 1&1&1\\2&2&2 \end{array}\right] \end{array}$ & $\begin{array}{l} L=9\\ M=2\\ N=3\\ \beta=1.5 \end{array}$\\
		\hline
	\end{tabular}
\end{center}
\end{table}

According to the above evaluations, the MPC and LB-LCD caching schemes are two special solutions of the joint optimization problem when $ \eta =1 $ and $ \eta =0$, respectively. The former can achieve the lowest cell average outage probability while the latter can achieve the minimum fronthaul usage. The proposed GA-based approach can achieve different tradeoffs between the cell average outage probability and fronthaul usage according to different weighting factors, which can achieve better performance than the MPC and LB-LCD schemes.

\subsection{Performances of the GA-based Approach and Mode Selection Approach}

The performances of the proposed GA-based approach and the mode selection approach are analyzed in this subsection. Besides the {MPC} and the {LB-LCD} caching schemes, two other widely used caching strategies are evaluated for comparison, one is \emph{random caching}, where each RRH caches the content files independently and randomly regardless of the files' popularity distribution, the other one is \emph{probabilistic caching}, where each RRH caches the files independently and randomly according to the files' popularity distribution, i.e., high-ranked files have higher probability to be cached \cite{delivery_content_journal, cost_balancing}. There are $ L=50 $ files \footnote{Alougth there is a huge amount of content files in practice, they can be classified into different categories \cite{MPC_original}, and the number of files in each category (or subcategory) is relatively limited, so the proposed algorithms can be performed on each category, the number of files evaluated in the simulation will not lose meaningful insights of the tradeoff caching optimization.}, $ N=7 $ RRHs with one RRH located at the cell center and the other 6 RRHs evenly distributed on the circle with radius $ {2R}/{3} $, and $ \beta = 1.5 $.

\begin{figure}[!t]
	{\includegraphics[width=\hsize]{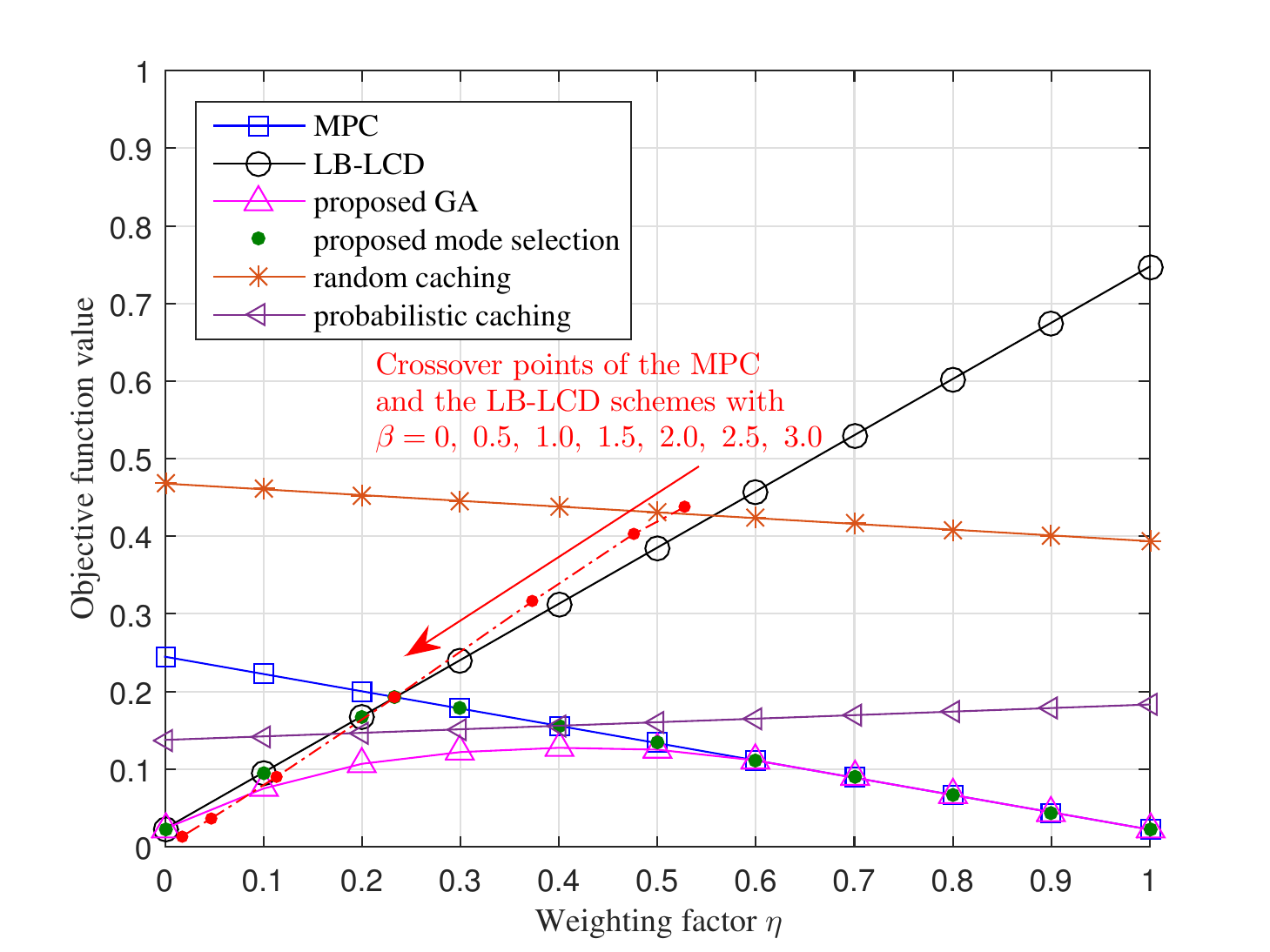}}
	\caption{Objective function value versus weighting factor $ \eta $. $L = 50, M = 5, N = 7 $.}
	\label{various_eta}
\end{figure}

Fig. \ref{various_eta} shows the objective value of different caching strategies with $ M=5 $. It can be seen from the figure that as the weighting factor $ \eta $ increases, i.e., more focus on minimization of outage probability, the objective value of MPC decreases linearly, while the objective value of the LB-LCD scheme increases linearly. The horizontal coordinate of the crossover point of the MPC and LB-LCD scheme ($ \eta_0 $) approaches  zero as the popularity skewness factor $ \beta $ increases. This is because that when $ \beta $ increases, the requesting probability $ P_l $ of the first few popular files increase significantly, then $ \sum\nolimits_{l=M+1}^{NM}P_l \rightarrow 0 $ in (\ref{eta_caculation_2}), thus $ \eta_0 \rightarrow 0 $. That is, as $ \beta $ increases, the MPC scheme will dominate with most values of $ \eta $. This can also be explained as follows. When $ \beta $ increases, the average fronthaul usage will depend more and more on the few files with higher ranks. These files can be cached in the RRHs under both of the MPC and the LB-LCD schemes, thus the MPC and the LB-LCD schemes are equivalent in terms of fronthaul usage, while the MPC can achieve lower outage probability. Therefore the MPC scheme is superior to the LB-LCD scheme. The crossover point $\eta_0 =0.23 $ when $ \beta =1.5 $ calculated through (\ref{eta_caculation_2}) exactly matches the simulation results. The above mentioned results are consistent with Proposition \ref{prop1}. 

The random caching strategy has a relative poor performance for all values of $ \eta $, which is because the files cached in the RRHs are selected randomly, there is neither a high probability to cache the same file for reducing the outage probability nor to cache different high-ranked files for reducing the fronthaul usage. While the probabilistic caching strategy can achieve better performance than the proposed mode selection approach in the middle range of $ \eta $, e.g., for $ \eta=0.2\sim 0.4 $, where both the cell average outage probability minimization and the fronthaul usage reduction are treated approximate equally. Which is because that, in probabilistic caching, each RRH will cache the high-ranked files with a higher probability, so there is a high probability for different RRHs to cache the same high-ranked files, which can reduce the cell average outage probability, and meanwhile the inherent randomness in the placement makes it possible to cache different files to reduce the fronthaul usage. It is also seen that the proposed GA-based approach can achieve better performance than the other caching strategies, for instance, the objective function value of the proposed GA algorithm is 18.25\% lower than a typical probabilistic caching scheme when $ \eta=0.4 $, and this improvement goes up to 87.9\% when $ \eta=1 $, the average improvement over all values of $ \eta $ is 47.5\%. 

\begin{figure}[!t]
	{\includegraphics[width=\hsize]{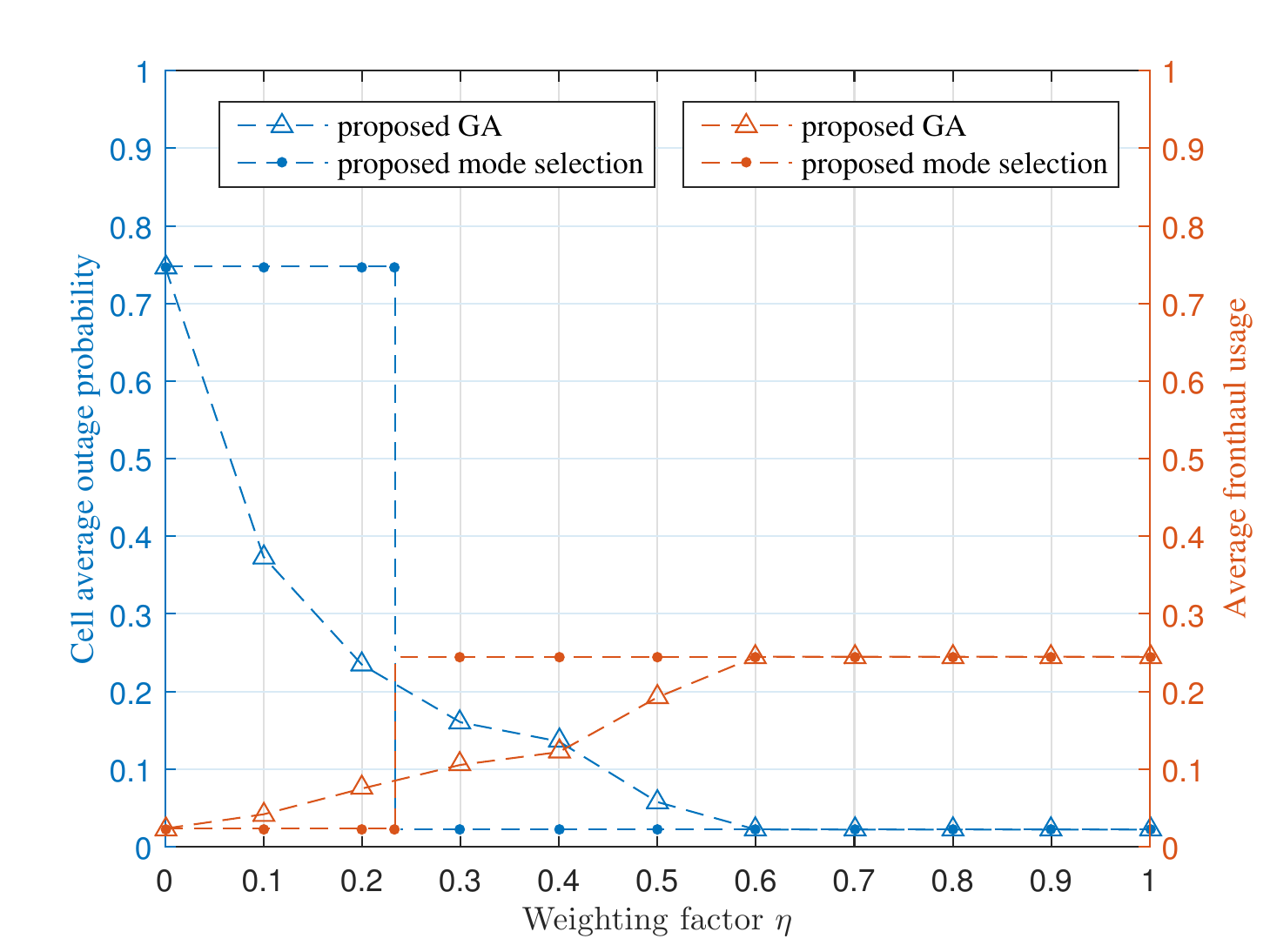}}
	\caption{Cell average outage probability and fronthaul usage versus weighting factor $ \eta $. $L = 50, M=5, N = 7, \beta=1.5 $.}
	\label{plotyy}
\end{figure}

Cell average outage probability and fronthaul usage of the proposed GA and the proposed mode selection approach versus weighting factor are shown in Fig. \ref{plotyy}. Note that the mode selection scheme is actually the LB-LCD scheme when $ \eta \leq \eta _0 $ and the MPC scheme when $ \eta > \eta _0 $, respectively. For the proposed GA approach, the solution is exactly the LB-LCD scheme when $ \eta = 0 $, as $ \eta $ increases, the cell average outage probability decreases and the fronthaul usage increases, and they reach the lower and upper bounds when $ \eta\geqslant 0.6 $, respectively, where the solution is the MPC scheme. The proposed GA approach can adjust the caching placement according to different weighting factors $ \eta $ while the mode selection scheme only chooses a caching placement between the MPC and the LB-LCD schemes based on whether $ \eta>\eta _0 $, so the proposed GA approach can achieve better performance than the mode selection scheme. However, the computational complexity of the mode selection scheme is extremely low.

\begin{figure}[!t]
	{\includegraphics[width=\hsize]{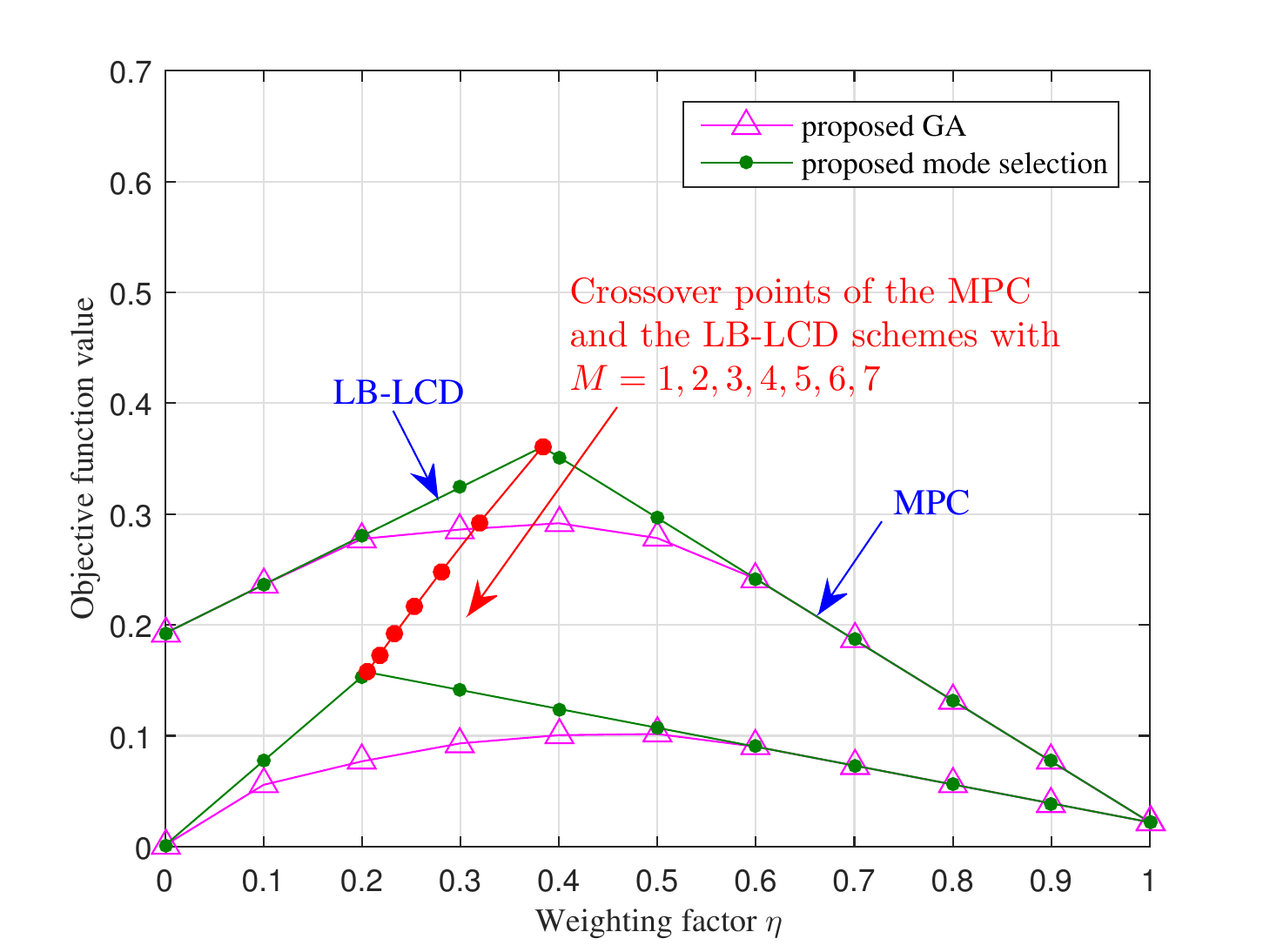}}
	\caption{Objective function value versus weighting factor $ \eta $. $L = 50, N = 7, \beta=1.5 $.}
	\label{various_M}
\end{figure}

Fig. \ref{various_M} shows the performance of the proposed GA and the mode selection scheme with different cache size $ M $. It can be seen from the figure that the mode selection scheme can achieve near-optimal performance over a wide range of the weighting factor $ \eta $. The vertex of the mode selection scheme, i.e., the crossover point of the MPC and the LB-LCD schemes moves toward the origin as the cache size $ M $ increases, i.e., the MPC scheme will dominate with most values of $ \eta $ as $ M $ increases. The reason is explained as follows. When $ M $ increases, more content files can be cached in the RRHs. The fronthaul usage depends mostly on the first few popular files cached in the RRHs, so the fronthaul usage will tend to be the same between the two schemes as $ M $ increases. The MPC scheme can achieve lower outage probability, further more, the cell average outage probability of the LB-LCD scheme increases as $ M $ increases, so the objective value of the MPC scheme will be much lower than that of the LB-LCD scheme, and the MPC scheme is superior to the LB-LCD scheme with most values of $ \eta $.

\begin{figure}[!t]
	{\includegraphics[width=\hsize]{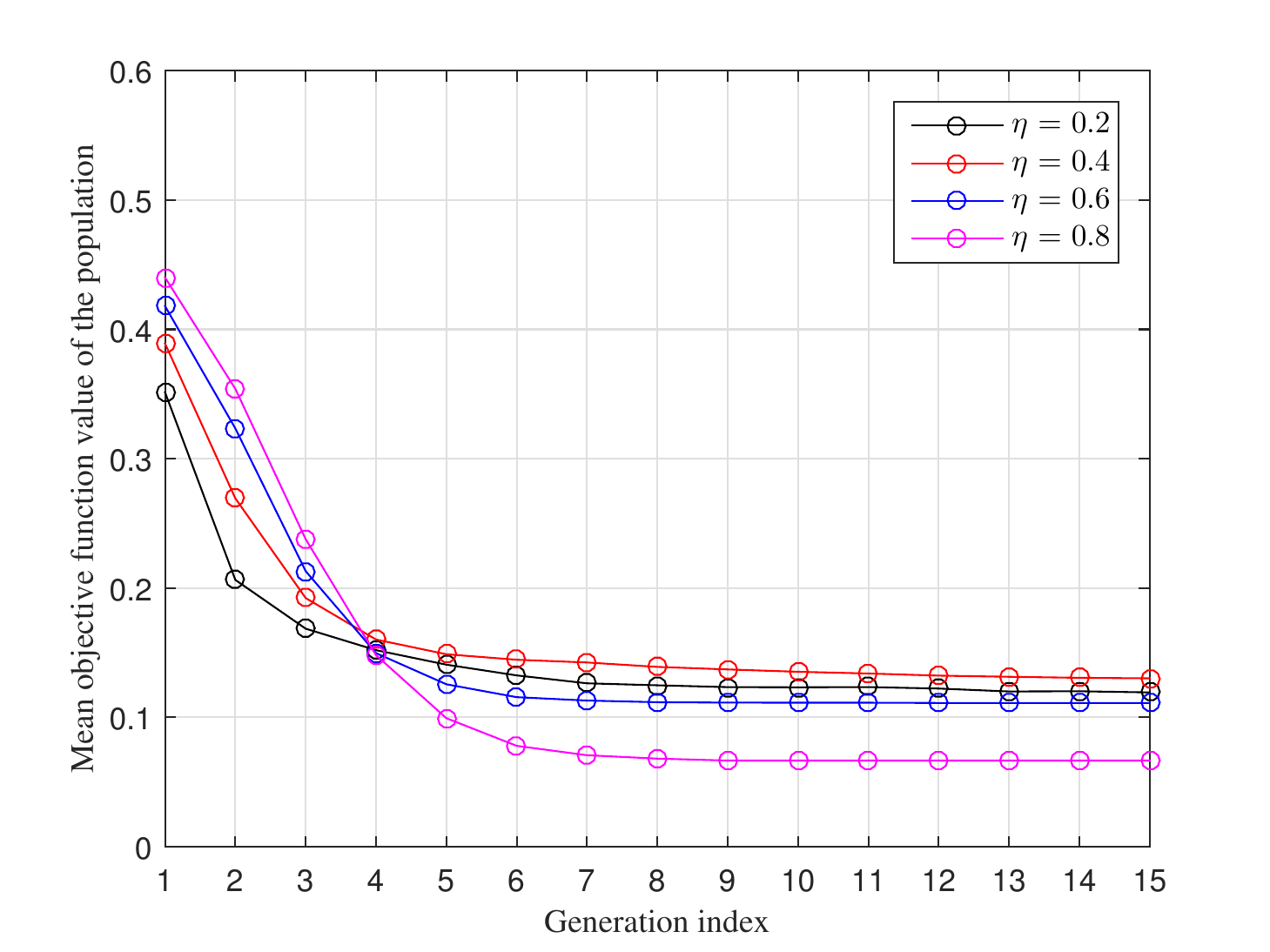}}
	\caption{Convergence behavior of the proposed GA approach. $ L = 50, M=5, N=7, \beta =1.5 $.}
	\label{converge}
\end{figure}

Fig. \ref{converge} shows the convergence behavior of the proposed GA approach. It can be seen from the figure that the mean objective value of the population converges within average 8 generations. The computational complexity is $ N_gN_p= 8 \times 50 = 400 $. While the computational complexity of the exhaustive search is $ \binom{50}{5}^7= 1.92\times 10^{44}$, which is not feasible in practice. As stated earlier, the popularity of the content files will remain the same for a relative long period, so the convergence behavior of the caching placement algorithm is not time-critical. Unlike the delivery stage, which needs to make an instant decision for coping with the dynamics of mobile networking systems, e.g., the rapid change of channel state information, it is not necessary for the caching strategy to make an instant decision due to the slow change of the statistics data (e.g., the request probabilities of the files). In addition, if parallel computing is adopted, the population size of the proposed GA approach can be increased without introducing additional execution time, while the converging speed of the GA will be accelerated. Thus, the GA approach can perform well with a satisfying converging speed in practice.

\section{Conclusion} \label{conclusion}

In this paper, we have investigated tradeoff caching strategy in Cloud-RAN for future mobile communications. In order to jointly minimize the cell average outage probability and fronthaul usage, the optimization problem is formulated as a weighted sum of the two objectives, with weighting factor $ \eta $ (and $ 1-\eta $). Analytical expressions of cell average outage probability and fronthaul usage have been presented and verified through simulations. Performances of two particular caching strategies have been analyzed, namely the MPC and the LB-LCD schemes. When the minimization of the cell average outage probability is more focused on, the MPC scheme is superior to the LB-LCD scheme, while the latter is superior to the former in the opposite situation, i.e., where the reduction of average fronthaul usage is more focused on. When the content files' popularity skewness factor $ \beta $ is larger, or the cache size of each RRH increase, the MPC scheme will dominate in a wide range of $ \eta $. Two heuristic approaches have been proposed to solve the joint optimization problem: one is the GA based approach which can achieve nearly the same optimal performance of exhaustive search, while the computational complexity is significantly reduced; the other is the mode selection approach with extremely low computational complexity, which can obtain near-optimal performance within a wide range of $ \eta $. Compared with a typical probabilistic caching scheme, the proposed GA approach can reduce the objective function value by up to 45.7\% on average and the proposed two mode selection caching strategy can provide an average improvement of 36.9\%. In practice, the RAN can make a decision of the tradeoff according to the system's statistics of fronthaul traffic and outage probability, and then adopt caching strategy through the proposed schemes.

\appendices 
\section{\label{A1}  Derivations of (\ref{the_pdf_snr}) and (\ref{the_cdf_snr}) }
\numberwithin{equation}{section}
\setcounter{equation}{0}

For a specific file $ F_l $, the subscript of file index $ l $ and the user's location $ x_0 $ are omitted without ambiguity. In (\ref{gamma_n}), $ |h_n|^2\sim\chi^2(2) $, and the PDF is given by \cite{digital_communication}
\begin{equation}
f_{|h_n|^2}(x) = \exp(-x),~x>0.
\end{equation}
Then the PDF of $ \gamma_n = \gamma_0S_n|h_n|^2 $ is
\begin{equation}
\label{gamma_pdf}
f_{\gamma_n} (\gamma)=\dfrac{1}{\gamma_0 S_n}\exp\left( {-\dfrac{\gamma}{\gamma_0 S_n}} \right),~\gamma >0,~n\in\Phi .
\end{equation}
The moment generation function (MGF) \cite{over_fading} of the random variable $\gamma_n$ is
\begin{equation}
\begin{aligned}
\label{the_MGF}
&M_{\gamma_n}(s)&=&\int_{0}^{\infty}f_{\gamma_n}(\gamma)e^{s\gamma}d\gamma&\\
& &=&\int_{0}^{\infty}\frac{1}{\gamma_0S_n}\exp\left(-\frac{\gamma}{\gamma_0S_n}\right)e^{s\gamma}d\gamma&\\
& &=&\frac{1}{1-\gamma_0S_n\cdot{s}}~,&
\end{aligned}
\end{equation}
and the range of convergence (ROC) is $ \operatorname{Re}\left(s\right)<\frac{1}{\gamma_0S_n} $.
Since the RRHs are distributed at different locations, $ \{\gamma_n, n\in\Phi\} $ is independent of each other, the MGF of received SNR $\gamma = \sum_{n\in \Phi}^{}\gamma_n $ is given by
\begin{equation}
\label{the_MGF_gamma}
M_{\gamma}(s) = \prod_{n\in \Phi}^{} M_{\gamma_n}(s)= \prod_{n\in \Phi}^{}\frac{1}{1-\gamma_0S_n\cdot{s}}~,\\
\end{equation}
and the ROC is $ {\bigcap}_{n\in\Phi}\operatorname{Re}\left(s\right)<\frac{1}{\gamma_0S_n} $.

Since there are $ I $ distinct distances $ d_1\neq d_2\neq \cdots \neq d_i\neq \cdots \neq d_I $ between the service RRHs and the user, and the $ i $-th distance has multiplicity of $ J_i $, (\ref{the_MGF_gamma}) can be rewritten as
\begin{equation}
M_{\gamma}(s)=\dfrac{1}{\left(1-\dfrac{1}{\lambda_1} s\right)^{J_1}\left(1-\dfrac{1}{\lambda_2} s\right)^{J_2}\cdots\left(1-\dfrac{1}{\lambda_I} s\right)^{J_I}}~,\\
\end{equation}
where $ \lambda_i = \frac{1}{\gamma_0Kd_i^{-\alpha}}, ~i\in \{1,2,\cdots,I\} $ is the $ i $-th pole of multiplicity $ J_i $ of $ M_{\gamma}(s) $, using partial fraction expansion, $ M_{\gamma}(s) $ can be expressed as

\begin{equation}
\label{MGF_a_nm}
M_{\gamma}(s)=\sum_{i=1}^{I}\sum_{j=1}^{J_i}{\dfrac{A_{ij}}{\left(1-\dfrac{1}{\lambda_i} s\right)^j}}~,
\end{equation}
where $ \{A_{ij}\} $ are the undetermined coefficients. Multiplying $ ( 1-\frac{1}{\lambda_i } s ) ^{J_i} $ to both sides of (\ref{MGF_a_nm}), then calculating the $ (J_i-j) $-th order derivate for both sides and let $ s=\lambda_i $, we have
\label{A_nm_caculate}
\begin{equation}
\begin{aligned}
&\dfrac{d^{J_i-j}}{ds^{J_i-j}}\left[ M_{\gamma}(s)\left( 1-\dfrac{1}{\lambda_i}s \right) ^{J_i}\right] \Bigg\vert_{s=\lambda_i} \\ 
=&\dfrac{d^{J_i-j}}{ds^{J_i-j}}\left. \left[ \sum_{i=1}^{I}\sum_{j=1}^{J_i}{\dfrac{A_{ij}}{\left(1-\dfrac{1}{\lambda_i}\cdot s\right)^j}}\left( 1-\dfrac{1}{\lambda_i}s \right) ^{J_i}\right] \right|_{s=\lambda_i} \\
=&(J_i-j)!\left(-\dfrac{1}{\lambda_i}\right)^{J_i-j}A_{ij}~.
\end{aligned}
\end{equation}
Thus $ A_{ij} $ is obtained as (\ref{the_A_nm}). 

The PDF of $ \gamma $ can be obtained by inversely transforming the MGF in (\ref{MGF_a_nm}).
Considering a general form of the PDF,
\begin{equation}
f(\gamma)=\gamma^ne^{-a\gamma},\quad \gamma\geq 0,
\end{equation}
where $ n\in \{0\}\cup\mathbb{Z}^+$, $ a\in \mathbb{R}^+ $. The MGF of $ f(\gamma) $ can be obtained by continuously using the method of integration by parts.
\begin{eqnarray}
& &M(s) \nonumber \\
&=&\int_{0}^{\infty}\gamma^ne^{-a\gamma}e^{s\gamma}d\gamma \nonumber \\
&=&-\dfrac{1}{a-s}\int_{0}^{\infty}\gamma^nde^{-(a-s)\gamma} \nonumber \\
&=&-\dfrac{1}{a-s}\left(\gamma^ne^{-(a-s)\gamma}\Big|_0^\infty -n\int_{0}^{\infty}e^{-(a-s)\gamma}\gamma^{n-1}d\gamma\right) \nonumber \\
& &\hspace{3cm}\vdots \nonumber \\
&=&\dfrac{n!}{(a-s)^{n+1}}~,
\end{eqnarray}
and the ROC is $\operatorname{Re}\left(s\right)<a $. Denote the pair of the PDF and its corresponding MGF as 
\begin{equation}
\label{pdf_mgf_pair}
f(\gamma)=\gamma^ne^{-a\gamma}\Longleftrightarrow M(s)=\dfrac{n!}{(a-s)^{n+1}}~.
\end{equation}
The CDF can be calculated in the same manner,
\begin{equation}
\label{cdf_pdf_pair}
\begin{aligned}
F(\gamma)&=\int_{0}^{\gamma}f(\gamma)d\gamma\\
&=\int_{0}^{\gamma}\gamma^n e^{-a\gamma}d\gamma\\
&=\dfrac{1}{a}\left[\dfrac{n!}{a^n}-\left(e^{-a\gamma}\sum_{k=0}^{n}\dfrac{n!}{(n-k)!a^k}\gamma^{n-k}\right)\right]~.\\
\end{aligned}
\end{equation}

According to (\ref{MGF_a_nm}) and (\ref{pdf_mgf_pair}), the PDF of the received SNR is obtained, as shown in (\ref{the_pdf_snr}). According to (\ref{cdf_pdf_pair}), the CDF of the received SNR is obtained as shown in (\ref{the_cdf_snr}).

\section{\label{A3} Proof of Proposition \ref{prop1} }

Without loss of generality, it is assumed that $ M_n = M, ~\forall n\in\mathcal{N},~|\mathcal{N}|>1 $. According to (\ref{p1}), it is obvious that the objective functions of the MPC and LCD schemes $ f_{obj}^{MPC} $ and $ f_{obj}^{LCD} $ are linearly (thus monotonic) continuous function of $ \eta $ on closed interval $ [0,1] $. 

When $ \eta = 0 $, $ f_{obj} =  \sum\nolimits_{l=1}^{L}P_lT_l  $. The objective values of the two schemes are
\begin{equation}
\begin{aligned}
f_{obj}^{MPC} &=&& \sum\limits_{l=1}^{M}P_l \underset{\substack{\uparrow \\ 0}}{{T_l}} &+& \sum\limits_{l=M+1}^{L}P_l\underset{\substack{\uparrow \\ 1}}{{T_l}} &=&\sum\limits_{l=M+1}^{L}P_l ~,\\
f_{obj}^{LCD}  &=&& \sum\limits_{l=1}^{NM}P_l\underset{\substack{\uparrow \\ 0}}{{T_l}} &+& \sum\limits_{l=NM+1}^{L}P_l\underset{\substack{\uparrow \\ 1}}{{T_l}} &=& \sum\limits_{l=NM+1}^{L}P_l~.
\end{aligned}
\end{equation}
Note that $ \sum\nolimits_{l=1}^{L}P_l=1 $ and $ P_1\geqslant P_2\geqslant\cdots\geqslant P_L $, where equality holds if and only if  $ \beta=0 $. Thus
\begin{equation}
\label{eta0}
f_{obj}^{MPC}\Big|_{\eta=0} > f_{obj}^{LCD}\Big|_{\eta=0} ~.
\end{equation}

When $ \eta = 1 $, $ f_{obj} =  \sum\limits_{l=1}^{L}P_l \mathbb{E}_{x_0}\left[ P_{out}^{(l)}(x_0) \right] $. Denoting $ \mathbb{E}_{x_0} \left[ P_{out}^{(l)}(x_0) \right]  $ as $ P_{cell,out}(l) $, then
\begin{equation}
\begin{aligned}
f_{obj}^{MPC} = & \underbrace{\sum\limits_{l=1}^{NM}P_lP^{MPC}_{cell,out}(l)}_C &+&\underbrace{\sum\limits_{l=NM+1}^{L}P_lP^{MPC}_{cell,out}(l)}_D~, \\ 
f_{obj}^{LCD} = & \underbrace{\sum\limits_{l=1}^{NM}P_lP^{LCD}_{cell,out}(l)}_E &+&\underbrace{\sum\limits_{l=NM+1}^{L}P_lP^{LCD}_{cell,out}(l)}_F~.
\end{aligned}
\end{equation}
According to the wireless transmission strategy, $ D = F $, where $ D $ and $ F $ correspond to the scenario that all the RRHs serve the user, while $ E > C $ because $ E $ denotes there is only one RRH serving the user, while $ C $ corresponds to all the RRHs serving the user. Thus
\begin{equation}
\label{eta1}
f_{obj}^{MPC}\Big|_{\eta=1} < f_{obj}^{LCD}\Big|_{\eta=1} ~.
\end{equation}
According to (\ref{eta0}), (\ref{eta1}) and the linearity of $ f^{MPC}_{obj} $ and $ f^{LCD}_{obj} $, there exists a crossover point $ \eta_0 \in [0,1]$ of the two objective functions. When $ \eta < \eta_0 $, the LCD scheme is superior to the MPC scheme, while when $ \eta>\eta_0 $, the MPC scheme is superior to the LCD scheme. 

Substituting $ \{a_{l,n}\} $ of the MPC and LCD schemes into (\ref{p1}), respectively, a linear equation of $ \eta $ is formulated, and the solution is shown as in (\ref{eta_caculation}). Because $ M = M_n,~\forall n $, (\ref{eta_caculation}) can be further written as (\ref{eta_caculation_2}).

The proof can be extended to the case that $ M_n $ is different with $ n $.


\ifCLASSOPTIONcaptionsoff
  \newpage
\fi

\bibliographystyle{IEEEtran}

\bibliography{reference_TVT}

\begin{thebibliography}{10}
\providecommand{\url}[1]{#1}
\csname url@samestyle\endcsname
\providecommand{\newblock}{\relax}
\providecommand{\bibinfo}[2]{#2}
\providecommand{\BIBentrySTDinterwordspacing}{\spaceskip=0pt\relax}
\providecommand{\BIBentryALTinterwordstretchfactor}{4}
\providecommand{\BIBentryALTinterwordspacing}{\spaceskip=\fontdimen2\font plus
\BIBentryALTinterwordstretchfactor\fontdimen3\font minus
  \fontdimen4\font\relax}
\providecommand{\BIBforeignlanguage}[2]{{%
\expandafter\ifx\csname l@#1\endcsname\relax
\typeout{** WARNING: IEEEtran.bst: No hyphenation pattern has been}%
\typeout{** loaded for the language `#1'. Using the pattern for}%
\typeout{** the default language instead.}%
\else
\language=\csname l@#1\endcsname
\fi
#2}}
\providecommand{\BIBdecl}{\relax}
\BIBdecl

\bibitem{myICC}
Z.~Ye, C.~Pan, H.~Zhu, and J.~Wang, ``Outage probability and fronthaul usage
  tradeoff caching strategy in {Cloud-RAN},'' in \emph{2017 IEEE International
  Conference on Communications (ICC)}, May 2017, pp. 1--6.

\bibitem{EE_5G}
S.~Buzzi, C.~L. I, T.~E. Klein, H.~V. Poor, C.~Yang, and A.~Zappone, ``A survey
  of energy-efficient techniques for {5G} networks and challenges ahead,''
  \emph{IEEE Journal on Selected Areas in Communications}, vol.~34, no.~4, pp.
  697--709, April 2016.

\bibitem{5G_survey2}
A.~Gupta and R.~K. Jha, ``A survey of {5G} network: architecture and emerging
  technologies,'' \emph{IEEE Access}, vol.~3, pp. 1206--1232, 2015.

\bibitem{trend_5G}
O.~Galinina, A.~Pyattaev, S.~Andreev, M.~Dohler, and Y.~Koucheryavy, ``{5G}
  multi-{RAT} {LTE-WiFi} ultra-dense small cells: performance dynamics,
  architecture, and trends,'' \emph{IEEE Journal on Selected Areas in
  Communications}, vol.~33, no.~6, pp. 1224--1240, June 2015.

\bibitem{green_hetnet}
K.~M.~S. Huq, S.~Mumtaz, J.~Bachmatiuk, J.~Rodriguez, X.~Wang, and R.~L.
  Aguiar, ``Green {HetNet} {CoMP}: energy efficiency analysis and
  optimization,'' \emph{IEEE Transactions on Vehicular Technology}, vol.~64,
  no.~10, pp. 4670--4683, October 2015.

\bibitem{cloud_overview}
A.~Checko, H.~L. Christiansen, Y.~Yan, L.~Scolari, G.~Kardaras, M.~S. Berger,
  and L.~Dittmann, ``Cloud {RAN} for mobile networks \textemdash a technology
  overview,'' \emph{IEEE Communications Surveys \& Tutorials}, vol.~17, no.~1,
  pp. 405--426, 2015.

\bibitem{centralization}
C.~L. I, J.~Huang, R.~Duan, C.~Cui, J.~Jiang, and L.~Li, ``Recent progress on
  {C-RAN} centralization and cloudification,'' \emph{IEEE Access}, vol.~2, pp.
  1030--1039, 2014.

\bibitem{comparison}
H.~Zhu, ``Performance comparison between distributed antenna and microcellular
  systems,'' \emph{IEEE Journal on Selected Areas in Communications}, vol.~29,
  no.~6, pp. 1151--1163, June 2011.

\bibitem{train}
J.~Wang, H.~Zhu, and N.~J. Gomes, ``Distributed antenna systems for mobile
  communications in high speed trains,'' \emph{IEEE Journal on Selected Areas
  in Communications}, vol.~30, no.~4, pp. 675--683, May 2012.

\bibitem{allocation_DAS}
H.~Zhu and J.~Wang, ``Radio resource allocation in multiuser distributed
  antenna systems,'' \emph{IEEE Journal on Selected Areas in Communications},
  vol.~31, no.~10, pp. 2058--2066, October 2013.

\bibitem{green_mimo}
C.~Pan, H.~Zhu, N.~J. Gomes, and J.~Wang, ``Joint precoding and {RRH} selection
  for user-centric green {MIMO} {C-RAN},'' \emph{IEEE Transactions on Wireless
  Communications}, vol.~16, no.~5, pp. 2891--2906, May 2017.

\bibitem{incomplete}
------, ``Joint user selection and energy minimization for ultra-dense
  multi-channel {C-RAN} with incomplete {CSI},'' \emph{IEEE Journal on Selected
  Areas in Communications}, vol.~35, no.~8, pp. 1809--1824, Aug 2017.

\bibitem{game_vt}
S.~C. Zhan and D.~Niyato, ``A coalition formation game for remote radio head
  cooperation in cloud radio access network,'' \emph{IEEE Transactions on
  Vehicular Technology}, vol.~66, no.~2, pp. 1723--1738, Feb 2017.

\bibitem{5G_vt}
R.~Yu, J.~Ding, X.~Huang, M.~T. Zhou, S.~Gjessing, and Y.~Zhang, ``Optimal
  resource sharing in {5G}-enabled vehicular networks: A matrix game
  approach,'' \emph{IEEE Transactions on Vehicular Technology}, vol.~65,
  no.~10, pp. 7844--7856, Oct 2016.

\bibitem{backhaul_challenge}
M.~Jaber, M.~A. Imran, R.~Tafazolli, and A.~Tukmanov, ``{5G} backhaul
  challenges and emerging research directions: a survey,'' \emph{IEEE Access},
  vol.~4, pp. 1743--1766, 2016.

\bibitem{redesigning}
J.~Liu, S.~Xu, S.~Zhou, and Z.~Niu, ``Redesigning fronthaul for next-generation
  networks: beyond baseband samples and point-to-point links,'' \emph{IEEE
  Wireless Communications}, vol.~22, no.~5, pp. 90--97, October 2015.

\bibitem{wireless_backhauling}
U.~Siddique, H.~Tabassum, E.~Hossain, and D.~I. Kim, ``Wireless backhauling of
  {5G} small cells: challenges and solution approaches,'' \emph{IEEE Wireless
  Communications}, vol.~22, no.~5, pp. 22--31, October 2015.

\bibitem{caching_air}
X.~Wang, M.~Chen, T.~Taleb, A.~Ksentini, and V.~C.~M. Leung, ``Cache in the
  air: exploiting content caching and delivery techniques for {5G} systems,''
  \emph{IEEE Communications Magazine}, vol.~52, no.~2, pp. 131--139, February
  2014.

\bibitem{caching_multicast}
K.~Poularakis, G.~Iosifidis, V.~Sourlas, and L.~Tassiulas, ``Exploiting caching
  and multicast for {5G} wireless networks,'' \emph{IEEE Transactions on
  Wireless Communications}, vol.~15, no.~4, pp. 2995--3007, April 2016.

\bibitem{fundamental}
M.~A. Maddah-Ali and U.~Niesen, ``Fundamental limits of caching,'' \emph{IEEE
  Transactions on Information Theory}, vol.~60, no.~5, pp. 2856--2867, May
  2014.

\bibitem{split}
J.~Duan, X.~Lagrange, and F.~Guilloud, ``Performance analysis of several
  functional splits in {C-RAN},'' in \emph{2016 IEEE 83rd Vehicular Technology
  Conference (VTC Spring)}, May 2016, pp. 1--5.

\bibitem{delivery_beamforming}
X.~Peng, J.~C. Shen, J.~Zhang, and K.~B. Letaief, ``Joint data assignment and
  beamforming for backhaul limited caching networks,'' in \emph{2014 IEEE 25th
  Annual International Symposium on Personal, Indoor, and Mobile Radio
  Communication (PIMRC)}, September 2014, pp. 1370--1374.

\bibitem{delivery_content_journal}
M.~Tao, E.~Chen, H.~Zhou, and W.~Yu, ``Content-centric sparse multicast
  beamforming for cache-enabled {Cloud RAN},'' \emph{IEEE Transactions on
  Wireless Communications}, vol.~15, no.~9, pp. 6118--6131, Sept. 2016.

\bibitem{match_cache}
F.~Pantisano, M.~Bennis, W.~Saad, and M.~Debbah, ``Match to cache: joint user
  association and backhaul allocation in cache-aware small cell networks,'' in
  \emph{2015 IEEE International Conference on Communications (ICC)}, June 2015,
  pp. 3082--3087.

\bibitem{letter}
H.~Hsu and K.~C. Chen, ``A resource allocation perspective on caching to
  achieve low latency,'' \emph{IEEE Communications Letters}, vol.~20, no.~1,
  pp. 145--148, January 2016.

\bibitem{delay}
X.~Li, X.~Wang, S.~Xiao, and V.~C.~M. Leung, ``Delay performance analysis of
  cooperative cell caching in future mobile networks,'' in \emph{2015 IEEE
  International Conference on Communications (ICC)}, June 2015, pp. 5652--5657.

\bibitem{prob_cl}
Y.~Zhou, Z.~Zhao, R.~Li, H.~Zhang, and Y.~Louet, ``Cooperation-based
  probabilistic caching strategy in clustered cellular networks,'' \emph{IEEE
  Communications Letters}, vol.~21, no.~9, pp. 2029--2032, Sept 2017.

\bibitem{optimizing_cl}
J.~Liao, K.~K. Wong, M.~R.~A. Khandaker, and Z.~Zheng, ``Optimizing cache
  placement for heterogeneous small cell networks,'' \emph{IEEE Communications
  Letters}, vol.~21, no.~1, pp. 120--123, Jan 2017.

\bibitem{pso}
S.~Wang, X.~Zhang, K.~Yang, L.~Wang, and W.~Wang, ``Distributed edge caching
  scheme considering the tradeoff between the diversity and redundancy of
  cached content,'' in \emph{2015 IEEE/CIC International Conference on
  Communications in China (ICCC)}, November 2015, pp. 1--5.

\bibitem{hyper}
Q.~Li, W.~Shi, X.~Ge, and Z.~Niu, ``Cooperative edge caching in
  software-defined hyper-cellular networks,'' \emph{IEEE Journal on Selected
  Areas in Communications}, vol.~35, no.~11, pp. 2596--2605, Nov 2017.

\bibitem{helpers}
J.~Song, H.~Song, and W.~Choi, ``Optimal caching placement of caching system
  with helpers,'' in \emph{2015 IEEE International Conference on Communications
  (ICC)}, June 2015, pp. 1825--1830.

\bibitem{backhaul-aware}
X.~Peng, J.~C. Shen, J.~Zhang, and K.~B. Letaief, ``Backhaul-aware caching
  placement for wireless networks,'' in \emph{2015 IEEE Global Communications
  Conference (GLOBECOM)}, December 2015, pp. 1--6.

\bibitem{two_tier}
Z.~Yan, S.~Chen, Y.~Ou, and H.~Liu, ``Energy efficiency analysis of
  cache-enabled two-tier {HetNets} under different spectrum deployment
  strategies,'' \emph{IEEE Access}, vol.~5, pp. 6791--6800, 2017.

\bibitem{prob_tvt}
Y.~Chen, M.~Ding, J.~Li, Z.~Lin, G.~Mao, and L.~Hanzo, ``Probabilistic
  small-cell caching: Performance analysis and optimization,'' \emph{IEEE
  Transactions on Vehicular Technology}, vol.~66, no.~5, pp. 4341--4354, May
  2017.

\bibitem{tier_level}
J.~Wen, K.~Huang, S.~Yang, and V.~O.~K. Li, ``Cache-enabled heterogeneous
  cellular networks: Optimal tier-level content placement,'' \emph{IEEE
  Transactions on Wireless Communications}, vol.~16, no.~9, pp. 5939--5952,
  Sept 2017.

\bibitem{zipf}
L.~Breslau, P.~Cao, L.~Fan, G.~Phillips, and S.~Shenker, ``Web caching and
  {Zipf-like} distributions: evidence and implications,'' in \emph{Eighteenth
  Annual Joint Conference of the IEEE Computer and Communications Societies
  (INFOCOM '99)}, vol.~1, March 1999, pp. 126--134.

\bibitem{FemtoCaching}
K.~Shanmugam, N.~Golrezaei, A.~G. Dimakis, A.~F. Molisch, and G.~Caire,
  ``Femtocaching: Wireless content delivery through distributed caching
  helpers,'' \emph{IEEE Transactions on Information Theory}, vol.~59, no.~12,
  pp. 8402--8413, Dec 2013.

\bibitem{chunk}
H.~Zhu and J.~Wang, ``Chunk-based resource allocation in {OFDMA} systems
  \textemdash part {I}: chunk allocation,'' \emph{IEEE Transactions on
  Communications}, vol.~57, no.~9, pp. 2734--2744, September 2009.

\bibitem{joint}
------, ``Chunk-based resource allocation in {OFDMA} systems \textemdash part
  {II}: joint chunk, power and bit allocation,'' \emph{IEEE Transactions on
  Communications}, vol.~60, no.~2, pp. 499--509, February 2012.

\bibitem{radio}
H.~Zhu, ``Radio resource allocation for {OFDMA} systems in high speed
  environments,'' \emph{IEEE Journal on Selected Areas in Communications},
  vol.~30, no.~4, pp. 748--759, May 2012.

\bibitem{multicriteria}
M.~Ehrgott, \emph{Multicriteria Optimazation}, 2nd~ed.\hskip 1em plus 0.5em
  minus 0.4em\relax Springer, 2005.

\bibitem{system_outage}
J.~Y. Wang, J.~B. Wang, M.~Chen, H.~M. Chen, X.~Dang, and H.~Y. Li, ``System
  outage probability analysis of uplink distributed antenna systems over a
  composite channel,'' in \emph{2011 IEEE 73rd Vehicular Technology Conference
  (VTC Spring)}, May 2011, pp. 1--5.

\bibitem{book}
R.~L. Burden and J.~D. Faires, \emph{Numerical Analysis}, 9th~ed.\hskip 1em
  plus 0.5em minus 0.4em\relax Brooks/Cole, Cengage Learning, 2011.

\bibitem{ga_survey}
M.~Srinivas and L.~M. Patnaik, ``Genetic algorithms: a survey,''
  \emph{Computer}, vol.~27, no.~6, pp. 17--26, June 1994.

\bibitem{introduction_GA}
M.~Mitchell, \emph{An Introduction to Genetic Algorithms}.\hskip 1em plus 0.5em
  minus 0.4em\relax Cambridge, MA, USA: MIT Press, 1998.

\bibitem{MPC_original}
H.~Ahlehagh and S.~Dey, ``Video-aware scheduling and caching in the radio
  access network,'' \emph{IEEE/ACM Transactions on Networking}, vol.~22, no.~5,
  pp. 1444--1462, Oct 2014.

\bibitem{LCD_original}
N.~Golrezaei, P.~Mansourifard, A.~F. Molisch, and A.~G. Dimakis, ``Base-station
  assisted device-to-device communications for high-throughput wireless video
  networks,'' \emph{IEEE Transactions on Wireless Communications}, vol.~13,
  no.~7, pp. 3665--3676, July 2014.

\bibitem{antenna}
E.~Park, S.~R. Lee, and I.~Lee, ``Antenna placement optimization for
  distributed antenna systems,'' \emph{IEEE Transactions on Wireless
  Communications}, vol.~11, no.~7, pp. 2468--2477, July 2012.

\bibitem{transmission_schemes}
H.~Kim, S.~R. Lee, K.~J. Lee, and I.~Lee, ``Transmission schemes based on sum
  rate analysis in distributed antenna systems,'' \emph{IEEE Transactions on
  Wireless Communications}, vol.~11, no.~3, pp. 1201--1209, March 2012.

\bibitem{DAS_randomness}
J.~Zhang and J.~G. Andrews, ``Distributed antenna systems with randomness,''
  \emph{IEEE Transactions on Wireless Communications}, vol.~7, no.~9, pp.
  3636--3646, September 2008.

\bibitem{cost_balancing}
A.~Alameer and A.~Sezgin, ``Resource cost balancing with caching in {C-RAN},''
  in \emph{2017 IEEE Wireless Communications and Networking Conference (WCNC)},
  March 2017, pp. 1--6.

\bibitem{digital_communication}
J.~G. Proakis and M.~Salehi, \emph{Digital Communications}, 5th~ed.\hskip 1em
  plus 0.5em minus 0.4em\relax London: McGraw-Hill, 2008.

\bibitem{over_fading}
M.~K. Simon and M.-S. Alouini, \emph{Digital Communication over Fading
  Channels}.\hskip 1em plus 0.5em minus 0.4em\relax Hoboken: John Wiley \&
  Sons, 2005.

\end{thebibliography}

%
%

\end{document}